\documentclass[draftcls,onecolumn,11pt]{IEEEtran}    
\usepackage[T1]{fontenc}

\usepackage{amsmath}
\usepackage{amssymb}
\usepackage{stmaryrd}
\usepackage{epic,eepic}
\usepackage{theorem}
\usepackage{pifont}  
\usepackage{euscript}
\usepackage{calc}
\usepackage[titles]{tocloft}

\usepackage[usenames]{color}          
\usepackage{xcolor}
\definecolor{Brown}{rgb}{0.55,0.0,0.10}
\definecolor{dgreen}{rgb}{0.00,0.56,0.00}
\definecolor{vertmoinsfonce}{rgb}{0.00,0.50,0.00}
\definecolor{vert}{rgb}{0.00,0.60,0.00}
\definecolor{llightggray}{rgb}{0.97,0.97,0.97}
\definecolor{lightggray}{rgb}{0.9,0.9,0.9}
\definecolor{ggray}{rgb}{0.5,0.5,0.5}
\definecolor{darkggray}{rgb}{0.25,0.25,0.25}
\definecolor{ddarkggray}{rgb}{0.1,0.1,0.1}
\definecolor{bleu}{rgb}{0.00,0.00,1.00}
\definecolor{darkblue}{rgb}{0,0,0.7}

\usepackage{epsf}           
\usepackage{graphicx}       
\usepackage{epsfig}         
\usepackage{pstricks}
\usepackage{pstricks-add}
\usepackage{pst-plot}
\usepackage{dsfont}
\usepackage{pst-circ}
\usepackage{pst-math}

\theoremheaderfont{\color{black}\normalfont\bfseries}

\newtheorem{lemma}{Lemma}
\newtheorem{theorem}{Theorem}[section]
\newtheorem{definition}[theorem]{Definition}
\newtheorem{corollary}[theorem]{Corollary}
\newtheorem{proposition}[theorem]{Proposition}

\theoremstyle{plain}{\theorembodyfont{\rmfamily}%
}
\theoremstyle{plain}{\theorembodyfont{\rmfamily}%
}
\theoremstyle{plain}{
\theorembodyfont{\rmfamily}

	\newtheorem{remark}[theorem]{Remark}

	}

\newcommand{\vex}{\ensuremath{\operatorname{vex}}}

\newcommand{\argmin}{\ensuremath{\operatorname{argmin}}}

\newcommand{\argmax}{\ensuremath{\operatorname{argmax}}}

\newcommand{\R}{\mathbb{R}}

\newcommand{\N}{\mathbb{N}}

\newcommand{\prob}{\mathbb{P}}

\newcommand{\E}{\mathbb{E}}
\newcommand{\Q}{\mathbb{Q}}

\font\dsrom=dsrom10 scaled 1200 \def \indic{\textrm{\dsrom{1}}}
\newcommand{\UN}{\indic}

\newcommand{\QQ}{\mathcal{Q}}

\newcommand{\PP}{\mathcal{P}}
\newcommand{\EE}{\mathcal{E}}
\newcommand{\mc}{\mathcal}

{\Large\normalsize}
{\Large\normalsize}

\usepackage[nomessages]{fp}
\newcommand\const[3][6]{%
    \edef\temporary{round(#3}%
    \expandafter\FPeval\csname#2\expandafter\endcsname
        \expandafter{\temporary:#1)}%
        \pstVerb{/#2 \csname#2\endcsname\space def}%
}

\newcommand{\PlotConditionalEntropy}[3]{
\def\qqA{#1 #3 mul  #1 #3 mul #1 neg 1 add #2 neg 1 add mul add div}
\def\qqB{#1 #3 neg 1 add mul #1 #3 neg 1 add mul #1 neg 1 add  #2 mul add div}
\def\ppA{x \qqA \space mul  #1 neg 1 add #2 neg 1 add mul #1 #3 mul add mul #1  x neg 1 add #2 neg 1 add mul x #3 mul add mul div} 
\def\ppB{x \qqB \space mul  #1 neg 1 add #2 mul #1 #3 neg 1 add mul add mul #1  x neg 1 add #2 mul x #3 neg 1 add mul add mul div}
\def\lambdaA{x neg 1 add #2 neg 1 add mul x #3 mul add}
\def\lambdaB{x neg 1 add #2 mul #3 neg 1 add x mul add}
\def\entropyA{\ppA \space ln 2 ln div \ppA \space mul neg 1 \ppA \space neg add ln 2 ln div 1 \ppA \space neg add mul neg add}
\def\entropyB{\ppB \space  ln 2 ln div \ppB \space mul neg 1 \ppB \space neg add ln 2 ln div 1 \ppB \space neg add mul neg add}
\def\HUZw{\lambdaA \space \entropyA \space mul \lambdaB\space  \entropyB \space mul add}
\psplot[plotpoints=100]{0.001}{0.999}{\HUZw}
}

\newcommand{\PlotGq}[3]{
\def\qqA{#1 #3 mul  #1 #3 mul #1 neg 1 add #2 neg 1 add mul add div}
\def\qqB{#1 #3 neg 1 add mul #1 #3 neg 1 add mul #1 neg 1 add  #2 mul add div}
\def\ppA{x \qqA \space mul  #1 neg 1 add #2 neg 1 add mul #1 #3 mul add mul #1  x neg 1 add #2 neg 1 add mul x #3 mul add mul div} 
\def\ppB{x \qqB \space mul  #1 neg 1 add #2 mul #1 #3 neg 1 add mul add mul #1  x neg 1 add #2 mul x #3 neg 1 add mul add mul div}
\def\lambdaA{x neg 1 add #2 neg 1 add mul x #3 mul add}
\def\lambdaB{x neg 1 add #2 mul #3 neg 1 add x mul add}
\def\entropyA{\ppA \space ln 2 ln div \ppA \space mul neg 1 \ppA \space neg add ln 2 ln div 1 \ppA \space neg add mul neg add}
\def\entropyB{\ppB \space  ln 2 ln div \ppB \space mul neg 1 \ppB \space neg add ln 2 ln div 1 \ppB \space neg add mul neg add}
\def\HUZw{\lambdaA \space \entropyA \space mul \lambdaB\space  \entropyB \space mul add}
\def\Gq{\HUZw \space neg 0.8813 add}
\psplot[linecolor=blue,plotpoints=100]{0.001}{0.145}{\Gq}
\psplot[linestyle=dashed,plotpoints=100]{0.145}{0.999}{\Gq}
}

\newcommand{\FigDynamicNew}[4]{

\const{qqA}{#1 * #3 / (#1 * #3 + (1 - #1) * (1 - #2))}
\const{qqB}{#1 * (1-#3) / (#1 * (1-#3) + (1 - #1) * #2)}
\const{nuA}{#4 * #1  * (1 - #2)/ (\qqA * (1 - #2) + (\qqA - #4) * #1 * (#2 + #3 - 1))}
\const{nuB}{#4 * #1  * #2/ (\qqB * #2 + (\qqB - #4) * #1 * (1 - #2 - #3))}

\const{pBnuA}{#4 * (1-#2)  * \qqB * ((1  -#1) * #2 + #1 * (1 - #3)) /((\qqA * (1 - #1) * (1 - #2) + (\qqA - #4) * #1 * #3 )* #2 + #4 * #1 * ( 1 - #2) * ( 1 - #3))}
\const{pAnuB}{#4 * #2  * \qqA * ((1  -#1) * (1 - #2) + #1 *  #3) /((\qqB * (1 - #1) * #2 + (\qqB - #4) * #1 * (1 - #3) )* (1 - #2) + #4 * #1 *  #2 *  #3)}

\def\ppA{x \qqA \space mul  #1 neg 1 add #2 neg 1 add mul #1 #3 mul add mul #1  x neg 1 add #2 neg 1 add mul x #3 mul add mul div} 
\def\ppB{x \qqB \space mul  #1 neg 1 add #2 mul #1 #3 neg 1 add mul add mul #1  x neg 1 add #2 mul x #3 neg 1 add mul add mul div}

\begin{figure}[!h]
\begin{center}
\psset{xunit=5cm,yunit=5cm}
\begin{pspicture}(-0.2,-0.28)(1.2,1.2)
\psline{->}(0,-0.1)(0,1.1)
\psline{->}(-0.1,0)(1.2,0)
\psline{-}(1,-0.1)(1,1.1)
\psline{-}(-0.1,1)(1.1,1)

\rput[r](-0.02,-0.05){$0$}
\rput[r](-0.02,1.04){$1$}
\rput[u](1.03,-0.05){$1$}
\rput[u](1.13,-0.05){$q$}

\rput[l]{-45}(#1,-0.06){$p_0 = #1$}
\psline[linestyle=dotted](#1,0)(#1,1)
\psline[linecolor=blue,linestyle=dotted]{-}(0,#1)(1,#1)
\rput[r](-0.02,#1){$p_0=\PP(u_1) = #1$}
\psline[linecolor=blue,linestyle=dotted]{-}(0,\qqA)(1,\qqA)
\rput[r](-0.02, \qqA){$\PP(u_1|z_0) = \qqA$}
\psline[linecolor=blue,linestyle=dotted]{-}(0,\qqB)(1,\qqB)
\rput[r](-0.02, \qqB){$\PP(u_1|z_1) = \qqB$}
\psdots(#1,\qqA)(#1,\qqB)(0,\qqA)(0,\qqB)(0,#1)(#1,#1)(#1,0)
\Aput*{$p_1$}
\psplot[plotpoints=100,linecolor=brown]{0}{1}{ \ppA}
\psplot[plotpoints=100,linecolor=brown]{0}{1}{ \ppB}
\psline{-}(0,0)(1,1)
\psline[linecolor=blue,linestyle=dashed]{-}(\nuA,0)(\nuA,1)
\psline[linecolor=blue,linestyle=dashed]{-}(\nuB,0)(\nuB,1)
\psline[linecolor=blue]{-}(0,#4)(1,#4)
\rput[l](1.04, #4){$\gamma = #4$}
\rput[l]{-45}(\nuA,-0.06){$\nu_0 = \nuA$}
\rput[l]{-45}(\nuB,-0.06){$\nu_1 = \nuB$}
\psdots[linecolor=blue](\nuA,0)(\nuB,0)(\nuA,#4)(\nuB,#4)(1,#4)
\psline[linecolor=red,linestyle=dashed]{-}(1,\pBnuA)(\nuA,\pBnuA)
\psline[linecolor=red,linestyle=dashed]{-}(1,\pAnuB)(\nuB,\pAnuB)
\psdots[linecolor=red](\nuA,\pBnuA)(\nuB,\pAnuB)(1,\pBnuA)(1,\pAnuB)
\rput[l](1.04, \pAnuB){$p_0(\nu_1) = \pAnuB$}
\rput[l](1.04, 0.97){$p_1(\nu_0) = \pBnuA$}

\rput[u](0.25, 0.6){$p_1(q)$}
\rput[u](0.72,0.45 ){$p_0(q)$}

\end{pspicture}
\caption{The posterior beliefs functions $p_0(q)$ and $p_1(q)$ defined in \eqref{eq:VeritableBelief0} and \eqref{eq:VeritableBelief1}, depending on the \textit{interim belief} $q\in [0,1]$, for $p_0=#1$, $\delta_1=#2$, $\delta_2=#3$ and $\gamma = #4$.}\label{fig:DynamicNew#1_#2_#3_#4}
\end{center}
\end{figure}}

\ifCLASSINFOpdf
\else
\fi

\begin{document}


\title{Strategic Communication with Side Information at the Decoder}

%

%
%
%


\author{\IEEEauthorblockN{Ma\"{e}l Le Treust\IEEEauthorrefmark{1} and 
Tristan Tomala \IEEEauthorrefmark{2}}\\
\IEEEauthorblockA{\IEEEauthorrefmark{1}
ETIS UMR 8051, Université Paris Seine, Université Cergy-Pontoise, ENSEA, CNRS,\\
6, avenue du Ponceau, 95014 Cergy-Pontoise CEDEX, FRANCE\\
Email: mael.le-treust@ensea.fr}\\
\thanks{\IEEEauthorrefmark{1} Maël Le Treust gratefully acknowledges financial support from INS2I CNRS, DIM-RFSI, SRV ENSEA, UFR-ST UCP, The Paris Seine Initiative and IEA Cergy-Pontoise.}
\IEEEauthorblockA{\IEEEauthorrefmark{2}
HEC Paris, GREGHEC UMR 2959\\
1 rue de la Libération, 78351 Jouy-en-Josas CEDEX, FRANCE\\
Email: tomala@hec.fr}
\thanks{\IEEEauthorrefmark{2} Tristan Tomala gratefully acknowledges the support of the HEC foundation and ANR/Investissements d'Avenir under grant ANR-11-IDEX-0003/Labex Ecodec/ANR-11-LABX- 0047.} \thanks{This work was presented in part at the 54th Allerton Conference, Monticello, Illinois, Sept. 2016 \cite{LeTreustTomala(Allerton)16}; the XXVI Colloque Gretsi, Juan-Les-Pins, France, Sept. 2017 \cite{LeTreustTomala(Gretsi)17}; the International Zurich Seminar on Information and Communication, Switzerland, Feb. 2018 \cite{LeTreustTomala(IZS)18}. This research has been conducted as part of the project Labex MME-DII (ANR11-LBX-0023-01). The authors would like to thank Institute Henri Poincaré (IHP) in Paris, France, for hosting numerous research meetings.}}


\vspace{-0.9cm}

\maketitle

%

\IEEEpeerreviewmaketitle





\vspace{-1.2cm}

\begin{abstract}

We investigate the problem of strategic point-to-point communication with side information at the decoder, in which the encoder and the decoder have mismatched distortion functions. The decoding process is not supervised, it returns the output sequence that minimizes the decoder's distortion function. The encoding process is designed beforehand and takes into account the decoder's distortion mismatch. When the communication channel is perfect and no side information is available at the decoder, this problem is referred to as {\em the Bayesian persuasion game} of Kamenica-Gentzkow in the Economics literature. We formulate the  strategic communication scenario as a joint source-channel coding problem with side information at the decoder. The informational content of the source influences the design of the encoding since it impacts differently the two distinct distortion functions. The  side information complexifies the analysis since the encoder is uncertain about the decoder's belief on the source statistics. We characterize the single-letter optimal solution by controlling the posterior beliefs induced by the Wyner-Ziv's source encoding scheme. This confirms the benefit of sending encoded data bits even if the decoding process is not supervised. 
\end{abstract}


\section{Introduction}\label{sec:Introduction}

What information should be communicated to a receiver who minimizes a mismatched distortion metric? This new question arises in the context of the internet of things (IoT) composed of a variety of devices which are able to interact and coordinate with each other in order to create new applications/services and reach their own goals. In this context,  wireless devices may have distinct objectives. For example, adjacent access points in crowded downtown areas, seeking to transmit at the same time, compete for the use of bandwidth; cognitive radio devices mitigate the interference effects by allocating their power budget over several parallel multiple access channels, as in \cite[Sec. IV]{LeTreustTomala(Allerton)16}. Such situations require new efficient techniques to coordinate communication traffic between  devices whose objectives are \emph{neither aligned, nor antagonistic}. This question differs from the classical paradigm in Information Theory which assumes that  communicating devices are of two types:  transmitters who pursue the common goal of transferring information; or opponents who try to mitigate the communication, e.g. the jammer corrupts the information, the eavesdropper infers it, the warden detects the covert transmission. In this work, we characterize the information-theoretic limits of  strategic communication between interacting autonomous devices having general distortion functions, not necessarily aligned.

\subsection{Scenario and contributions}

 We formulate the \emph{strategic communication problem} as a joint source-channel coding problem with decoder's side information, in which the encoder and the decoder are endowed with distinct distortion functions $d_{\textsf{e}}$ and $d_{\textsf{d}}$. Both distortion functions depend on the symbols of source, side information and decoder's outputs. We consider that the decoder is not supervised and strategic, \textit{i.e.} it selects a decoding strategy $\tau$ which is optimal for its own distortion function. The encoder anticipates the mismatch of the decoder's distortion's and implements an encoding strategy $\sigma$ that minimizes its distortion. The problem we consider lies on the bridge between Information Theory and Game Theory, and is given by
\begin{align}
&\lim_{n\to+\infty}\inf_{\sigma}\max_{{\tau} \in \argmin_{\tilde{\tau}} d_{\textsf{d}}^{\,n}(\sigma, \tilde{\tau})} d_{\textsf{e}}^{\,n}(\sigma, {\tau} ). \label{eq:GameProblem000}
\end{align}

The main difficulty comes from the fact that the decoder can possibly choose an output sequence that induces a catastrophic distortion for the encoder. This modifies the encoder's objective, which is to control the decoder's posterior beliefs regarding source symbols, rather than transferring information. The closest paper in the literature is \cite{LeTreustTomala19}, in which no side information is available at the decoder. 

In this article, we demonstrate that an optimal strategy of the decoder produces a sequence of outputs which is almost the same as the one prescribed by the Wyner-Ziv's coding in \cite{wyner-it-1976}, adapted to the joint source-channel scenario by Merhav-Shamai in \cite{MerhavShamai03}. This demonstrates that the Wyner-Ziv's source coding reveals nothing but the \emph{exact} amount of information needed by the decoder. We establish a single-letter expression for the decoder's posterior belief, that allows us to characterize the solution of problem \eqref{eq:GameProblem000}. Our solution boils down to the previous results in \cite{wyner-it-1976} and \cite{MerhavShamai03}, when both distortion functions are equal.

Then, we reformulate the solution in terms of a convex closure of an auxiliary distortion function with an entropy constraint. This simplifies the optimization over the set of  probability distributions by reducing the dimension of the single-letter problem. This second solution also relates our result to the literature on Bayesian persuasion games, see \cite{KamenicaGentzkow11}. As an illustration, we consider the example of the doubly symmetric binary source introduced by Wyner-Ziv in \cite[Sec. II]{wyner-it-1976}, for which we compute the optimal solution explicitly. We notice that the optimal cardinality of the auxiliary random variable is either two or three, depending on the source and channel parameters.

We point out three essential features of  strategic communication problem with decoder's side information. 
\begin{itemize}
\item[1.] Each source symbol has a different impact on the encoder and the decoder's distortion functions, hence it is optimal to encode each symbol differently.
\item[2.] The noiseless version of this problem without decoder's side information  corresponds to the Bayesian persuasion game of Kamenica-Gentzkow \cite{KamenicaGentzkow11}. In that case, the optimal information disclosure policy requires a fixed amount of information bits. When the channel capacity is larger than this amount, it is optimal not to use all the channel resource.
\item[3.] The decoder's side information has two opposite effects on the optimal encoder's distortion: it enlarges the set of decoder's posterior beliefs, so it may decrease the encoder's distortion; it reveals partial information to the decoder, so it forces some decoder's best-reply symbols which might be sub-optimal for the encoder's distortion.
\end{itemize}

\subsection{Related literature}

 The problem of ``strategic communication'' in information-theoretic setting has been formulated by Akyol \textit{et al.} in \cite{AkyolLangbortBasar15},  \cite{AkyolLangbortBasar16}, \cite{AkyolLangbortBasarIEEE17}. The authors characterize the optimal solution for Gaussian source, side information and channel, with the Crawford-Sobel's quadratic cost functions \cite{CrawfordSobel1982StrategicInformation}. They prove that the optimal solution in the one-shot problem is also optimal when considering several strategic communication problems. This is not the case for general discrete source, channel and mismatched distortion functions. These results were further extended in \cite{NadendlaLangbortBasar(TCOM)18} for non-identical prior beliefs about the source and the channel. The problem of strategic communication was introduced in the Control Theory literature by Sar{\i}ta\c{s} \textit{et al.} in \cite{SaritasYukselGezici(ArXiV)17} and  \cite{SaritasYukselGezici(TAC)17}. The authors extend the model of Crawford-Sobel to multidimentional sources and noisy channels and they determine whether  the optimal policies are linear or based on  some quantization. The connection to the binary hypothesis-testing problem was pointed out in \cite{SaritasGeziciYuksel(TSP)19}. Sender-receiver games are also investigated in \cite{FarokhiTeixeiraLangbort17}, for the problem of ``strategic estimation'' involving self-interested sensors; and in \cite{MarecekShortenYu15}, \cite{TavafoghiTeneketzis(Allerton)17}, for the ``network congestion'' problem. In \cite{DughmiXu16}, \cite{Dughmi17}, \cite{DughmiKempeQiang16}, the authors investigate the computational aspects of the Bayesian persuasion game, when the signals are noisy. In \cite{BerryTse(ShannonMetNash)11}, \cite{PerlazaTandonPoorHan15}, the interference channel coding problem is formulated  as a game in which the users, i.e. the pairs of encoder/decoder, are allowed to use \emph{any} encoding/decoding strategy. The authors compute the set of Nash equilibria for linear deterministic and Gaussian channels. The non-aligned devices' objectives are captured by distinct distortion functions. Coding for several distortion measures is investigated for ``multiple descriptions coding'' in  \cite{GamalCover82}, for the lossy version of ``Steinberg's common reconstruction'' problem in \cite{LapidothMalarWigger14}, for the problem of minimax distortion redundancy in \cite{DemboWeissman03}, for ``lossy broadcasting'' in \cite{TimoGrantKramer13}, for an alternative measure of ``secrecy'' in \cite{Yamamoto88}, \cite{Yamamoto97}, \cite{SchielerCuff(RateDistortion14)}, \cite{SchielerCuff(Henchman)16}. 

The lossy source coding problem with mismatched distortion functions was formulated by Lapidoth, in \cite{Lapidoth97}. In this model, the decoder attempts to reconstruct a source sequence that was encoded with respect to another distortion metric. The problem of the mismatch channel capacity was studied in \cite{MerhavKaplanLapidothShamai94}, \cite{Zamir02}, \cite{Somekh15}, \cite{ScarlettMartinezFabregas16}, \cite{ZhouTanMotani19}, in which the decoding metric is not necessarily matched with the channel statistics.

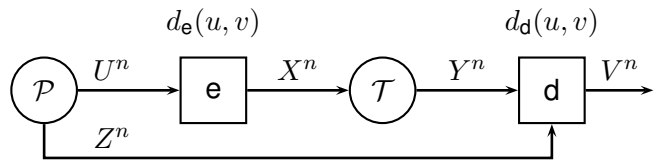
\begin{figure}[!ht]
\begin{center}
\psset{xunit=0.9cm,yunit=0.9cm}
\begin{pspicture}(0,-1.1)(8.5,1.7)
\pscircle(0,0.5){0.45}
\psframe(2,0)(3,1)
\pscircle(5,0.5){0.45}
\psframe(7,0)(8,1)
\psline[linewidth=1pt]{->}(0.5,0.5)(2,0.5)
\psline[linewidth=1pt]{->}(3,0.5)(4.5,0.5)
\psline[linewidth=1pt]{->}(5.5,0.5)(7,0.5)
\psline[linewidth=1pt]{->}(8,0.5)(9,0.5)
\psline[linewidth=1pt]{->}(0,0)(0,-0.5)(7.5,-0.5)(7.5,0)
\rput[u](1,-0.2){$Z^{n}$}
\rput[u](1,0.8){$U^{n}$}
\rput[u](3.75,0.8){$X^n$}
\rput[u](6.25,0.8){$Y^n$}
\rput[u](8.5,0.8){$V^n$}
\rput(0,0.5){$\PP$}
\rput(5,0.5){$\mc{T}$}
\rput(2.5,0.5){$\textsf{e}$}
\rput(7.5,0.5){$\textsf{d}$}
\rput(2.5,1.5){$d_{\textsf{e}}(u,v)$}
\rput(7.5,1.5){$d_{\textsf{d}}(u,v)$}
\end{pspicture}
\caption{The information source $U$ and side information $Z$ are drawn i.i.d. according to $\PP_{UZ}$ and the channel $\mc{T}_{Y|X}$ is memoryless. The encoder $\textsf{e}$ and the decoder $\textsf{d}$ minimize mismatched distortion functions $d_{\textsf{e}}(u,v) \neq d_{\textsf{d}}(u,v)$. }
\label{fig:StrategicEmpiricalCoordination}
\end{center}
\end{figure}

The problem of ``strategic information transmission'' has been well studied in the Economics literature since the seminal paper by Crawford-Sobel \cite{CrawfordSobel1982StrategicInformation}. In this model, a better-informed sender transmits a signal to a receiver, who takes an action which impacts both sender and receiver's utility functions. The problem consists in determining the  \emph{optimal information disclosure} policy given that the receiver' best-reply action affects the sender's utility, see \cite{Forges94} for a survey.  In \cite{KamenicaGentzkow11}, Kamenica-Gentzkow introduced the Bayesian persuasion game in which the sender \emph{commits} to an information disclosure policy before the game starts. This subtle change of rules of the game induces a very different equilibrium solution related to \emph{Stackelberg equilibrium} \cite{stackelberg-book-1934}, instead of \emph{Nash equilibrium} \cite{Nash51}. This problem was later referred to as ``information design'' in \cite{BergemannMorris16}, \cite{Taneva16}, \cite{BergemannMorris17} and extended to the setting with ``heterogeneous beliefs'' in \cite{AlonsoCamara(JET)2016} and \cite{LaclauRenou17}. In most of the articles in the Economics literature, the transmission between the sender and the receiver is noise-free; except in \cite{TsakasTsakas2017}, \cite{Blume}, \cite{HernandezVonStengel14} where the noisy transmission is investigated in a finite block-length with no-error regime. 
Interestingly, Shannon's mutual information is widely accepted as a \emph{cost of information} for the problem of ``rational inattention'' in \cite{Sims03} and for the problem of ``costly persuasion'' in \cite{GentzkowKamenica14}, without explicit reference to a coding problem.

Entropy and mutual information appear endogenously in repeated games with finite automata and bounded recall \cite{NeymanOkada99}, \cite{NeymanOkada00}, \cite{NeymanOkada09}, with private observation \cite{GossnerVieille02}, or with imperfect monitoring \cite{GossnerTomala06}, \cite{GossnerTomala07}, \cite{GossnerLarakiTomala09}. In \cite{GossnerHernandezNeyman06}, the authors investigate a sender-receiver game with common interests by formulating a coding problem. They characterize the optimal solution via the mutual information. This result was later refined by Cuff in  \cite{Cuff(ImplicitCoordination)11} and referred to as the  ``coordination problem'' in \cite{CuffPermuterCover10}, \cite{CuffSchieler11}, \cite{LeTreust(EmpiricalCoordination)17}, \cite{CerviaLuzziLeTreustBloch(IT)18}. 


 The paper is organized as follows. The strategic communication problem is formulated in Sec. \ref{sec:StrategicCoding}. The encoding and decoding strategies and the distortion functions are defined in Sec. \ref{sec:StrategyUtility}. The strategic communication scenario is introduced in Sec. \ref{sec:PersuasionGame}. Our coding result and the four different characterizations are stated in Sec. \ref{sec:Characterization}. The first one is a linear program under an information constraint, formulated in Sec. \ref{sec:LinearProgram}. The main Theorem is stated in Sec.\ref{sec:MainResult}, and the sketch of proof is in Sec. \ref{sec:SketchProof}. In Sec. \ref{sec:Concavification}, we reformulate the solution in terms of three different convex closures. Sec. \ref{sec:ExampleBinary} provides an example based on a binary source, binary side information and binary decoder's actions. The proofs are stated in App \ref{sec:ProofConcavification} - \ref{sec:ProofPropDSBS}.

\section{Strategic communication problem}\label{sec:StrategicCoding}

\subsection{Coding strategies and distortion functions}\label{sec:StrategyUtility}

We denote by $\mc{U}$, $\mc{Z}$, $\mc{X}$, $\mc{Y}$, $\mc{V}$ the finite sets of information source, decoder's side information, channel inputs, channel outputs and decoder's outputs. Uppercase letters $U^n=(U_1,\ldots,U_n)\in\mc{U}^n$ and $Z^n$, $X^n$, $Y^n$, $V^n$ stand for sequences of random variables, whereas lowercase letters $u^n=(u_1,\ldots,u_n)\in\mc{U}^n$ and $z^n$, $x^n$, $y^n$, $v^n$ stand for sequences of realizations. We denote by $\Delta(\mc{X})$ the set of  probability distributions over $\mc{X}$, i.e. the probability simplex. For a probability distribution $\QQ_{X}\in  \Delta(\mc{X})$, we write $\QQ(x)$ instead of $\QQ_X(x)$ for the probability value assigned to realization $x\in \mc{X}$. The notation $\QQ_{X}(\cdot|y)\in \Delta(\mc{X})$ denotes the conditional probability distribution of $X\in\mc{X}$ given the realization $y\in \mc{Y}$ and $\QQ_{X}^{\otimes n}\in \Delta(\mc{X}^n)$ denotes the i.i.d. probability distribution. The distance between two probability distributions $\QQ_X$ and $\PP_X$ is based on $L^1$ norm, denoted by $||\QQ_X - \PP_X||_{1}= \sum_{x\in\mc{X}} |\QQ(x) - \PP(x)|$. We denote by $D(\QQ_X||\PP_X)$ the K-L (Kullback-Leibler) divergence. The notation $U  -\!\!\!\!\minuso\!\!\!\!-X    -\!\!\!\!\minuso\!\!\!\!-  Y$ stands for the Markov chain property corresponding to $\PP_{Y|XU} = \PP_{Y|X}$. We consider an i.i.d. information source and a memoryless channel distributed according to $\PP_{UZ}\in \Delta(\mc{U}\times \mc{Z})$ and $\mc{T}_{Y|X} : \mc{X} \to \Delta(\mc{Y})$, as depicted in Fig. \ref{fig:StrategicEmpiricalCoordination}.

\begin{definition}[Encoding and decoding strategies]\label{def:Code}$\;$\\
The encoding strategy $\sigma$ and the decoding strategy $\tau$ are defined by
\begin{align}
\sigma& : \mc{U}^{n} \longrightarrow \Delta(\mc{X}^n)  ,\label{eq:EncodingFunction}\\
\tau& : \mc{Y}^n \times   \mc{Z}^n  \longrightarrow  \Delta( \mc{V}^n)  . \label{eq:DecodingFunction}
\end{align}
Both strategies $(\sigma,\tau)$ are stochastic and induce a joint probability distribution $\PP_{\sigma,\tau} \in\Delta(\mc{U}^{n} \times\mc{Z}^{n} \times\mc{X}^{n}\times\mc{Y}^{n}  \times\mc{V}^{n} )$ over the $n$-sequences of symbols, defined by
\begin{align}
 \PP_{\sigma,\tau}=&\bigg(\prod_{t=1}^n\PP_{U_tZ_t} \bigg)\sigma_{X^n|U^n}
   \bigg(\prod_{t=1}^n \mc{T}_{Y_t|X_t} \bigg) \tau_{V^n|Y^nZ^n},
\end{align}
where $\sigma_{X^n|U^n}$ and $\tau_{V^n|Y^nZ^n}$ denote to the  conditional probability distributions induced by the strategies $\sigma$ and $\tau$.
\end{definition}
The encoding and decoding strategies $(\sigma,\tau)$ are defined in the same way as for the {joint source-channel coding problem with side information at the decoder} studied in \cite{MerhavShamai03}, based on Wyner-Ziv's setting in \cite{wyner-it-1976}. Unlike these previous works, we assume that the encoder and the decoder minimize distincts distortion functions. 

\begin{definition}[Distortion functions]\label{def:Utilities} 
The single-letter distortion functions of the encoder and decoder are defined by
\begin{align}
d_{\textsf{e}} : \mc{U} \times \mc{Z} \times \mc{V} \longrightarrow \R,\\
d_{\textsf{d}} : \mc{U} \times \mc{Z} \times \mc{V} \longrightarrow \R.
\end{align}
The long-run distortion functions $d_{\textsf{e}}^{\,n}(\sigma,\tau)$ and $d_{\textsf{d}}^{\,n}(\sigma,\tau)$ are evaluated with respect to the probability distribution $\PP_{\sigma,\tau}$ induced by the strategies $(\sigma,\tau)$
\begin{align}
d_{\textsf{e}}^{\,n}(\sigma,\tau) =& \E_{\sigma,\tau} \Bigg[ \frac{1}{n} \sum_{t=1}^n d_{\textsf{e}}(U_t,Z_t,V_t) \Bigg] \nonumber\\
=& \sum_{u^n,z^n,v^n}\PP_{\sigma,\tau}\big(u^n,z^n,v^n \big) \cdot  \Bigg[  \frac{1}{n} \sum_{t=1}^n d_{\textsf{e}}(u_t,z_t,v_t) \Bigg],\\
d_{\textsf{d}}^{\,n}(\sigma,\tau) =& \sum_{u^n,z^n,v^n}\PP_{\sigma,\tau}\big(u^n,z^n,v^n \big) \cdot  \Bigg[  \frac{1}{n} \sum_{t=1}^n d_{\textsf{d}}(u_t,z_t,v_t) \Bigg].
\end{align}
\end{definition}

\subsection{Strategic communication scenario}\label{sec:PersuasionGame}

In this work, the encoder and the decoder are autonomous devices that choose the encoding strategy $\sigma$ and the decoding strategy $\tau$ in order to minimize their long-run distortion $d_{\textsf{e}}^{\,n}(\sigma,\tau)$ and $d_{\textsf{d}}^{\,n}(\sigma,\tau)$. We assume that the strategic communication takes place as follows:
\begin{itemize}
\item[$\bullet$] Before the transmission starts, the encoder chooses the strategy $\sigma$ and announces it to the decoder. 
\item[$\bullet$] The sequences $(U^n, Z^n,X^n,Y^n)$ are drawn according to the joint probability distribution $\Big(\prod_{t=1}^n\PP_{U_tZ_t} \Big) \sigma_{X^n|U^n} \Big(\prod_{t=1}^n\mc{T}_{Y_t|X_t}\Big)$.
\item[$\bullet$] The decoder knows $\sigma$, observes the sequences of symbols $(Y^n,Z^n)$, and is free to choose any decoding strategy $\tau$, in order to return a sequence of symbols $V^n$.
\end{itemize}

This setting corresponds to the Bayesian persuasion game \cite{KamenicaGentzkow11}, in which the encoder commits to an \emph{information disclosure policy} $\sigma$, and the decoder chooses a decoding strategy $\tau$ accordingly.

\begin{definition}[Decoder's Best-Replies]\label{def:BestReplyStrat} 
For any encoding strategy $\sigma$, the set of best-reply decoding strategies $ \textsf{BR}_{\textsf{d}}(\sigma)$ is defined by
\begin{align}
 \textsf{BR}_{\textsf{d}}(\sigma) =&\argmin_{\tau} d_{\textsf{d}}^{\,n}(\sigma, \tau) = \bigg\{\tau ,\text{ s.t. } \;  d_{\textsf{d}}^{\,n}(\sigma, \tau) \leq d_{\textsf{d}}^{\,n}(\sigma, \widetilde{\tau})  , \; \forall \widetilde{\tau} \neq \tau \bigg\}.
\end{align}
\end{definition}
In case there are several best-reply strategies, we assume that the decoder chooses the one that maximizes the encoder's distortion  $\max_{\tau \in \textsf{BR}_{\textsf{d}}(\sigma)} d_{\textsf{e}}^{\,n}(\sigma, \tau)$, so that the solution is robust to the exact specification of decoder's strategy. 

We aim at characterizing the asymptotic behavior of
\begin{align}
\inf_{\sigma}\max_{\tau \in \textsf{BR}_{\textsf{d}}(\sigma)} d_{\textsf{e}}^{\,n}(\sigma, \tau). \label{eq:GameProblem}
\end{align}

The decoding process $\tau$ is not supervised, it is strategic, causing the mismatch of the sequence of decoder's outputs. The design of the encoding strategy $\sigma$ anticipates this mismatch. Does the ``strategic decoder'' necessarily decode the coded bits of information? We provide a positive answer in Theorem \ref{theo:MaxMinStackelberg}, by refining the analysis of the Wyner-Ziv's encoding scheme  \cite{wyner-it-1976}. More precisely, we show that the symbols induced by Wyner-Ziv's decoding $\tau^{\textsf{wz}}$ coincide with those induced by any best-reply $\tau \in\textsf{BR}_{\textsf{d}}(\sigma^{\textsf{wz}})$ to Wyner-Ziv's encoding $\sigma^{\textsf{wz}}$, for a large fraction of stages. 



\begin{remark}[Stackelberg v.s. Nash equilibrium]
The optimization problem in \eqref{eq:GameProblem} corresponds to a Stackelberg equilibrium \cite{stackelberg-book-1934} in which the encoder is the leader and the decoder is the follower, unlike the Nash equilibrium  \cite{Nash51}  in which the two devices choose their strategy simultaneously. 
\end{remark}

\begin{remark}[Equal distortion functions]
When the encoder and decoder have equal distortion functions $d_{\textsf{e}} =d_{\textsf{d}}$, the problem in \eqref{eq:GameProblem} boils down to the problem studied by Merhav-Shamai in \cite{MerhavShamai03}, in which both strategies $(\sigma,\tau)$ are chosen \emph{jointly}, in order to minimize a distortion function
\begin{align}
\inf_{\sigma}\max_{\tau \in \textsf{BR}_{\textsf{d}}(\sigma)} d_{\textsf{e}}^{\,n}(\sigma, \tau) = 
\inf_{\sigma}\min_{\tau} d_{\textsf{e}}^{\,n}(\sigma, \tau) = 
\min_{(\sigma,\tau)} d_{\textsf{e}}^{\,n}(\sigma, \tau), \label{eq:GameProblemEqual}
\end{align}
since by Definition \ref{def:BestReplyStrat}, $\tau \in \textsf{BR}_{\textsf{d}}(\sigma) \Longleftrightarrow d_{\textsf{d}}^{\,n}(\sigma, \tau) = \min_{\tau'} d_{\textsf{d}}^{\,n}(\sigma, \tau') $. For such scenario, Merhav-Shamai's separation result in \cite{MerhavShamai03}, shows that it is optimal to concatenate Wyner-Ziv's source coding with Shannon's channel coding.
\end{remark}


\section{Characterizations}\label{sec:Characterization}

\subsection{Linear program with an information constraint}\label{sec:LinearProgram}

We define the encoder's optimal distortion  $D_{\textsf{e}}^{\star}$. 
\begin{definition}[Target distributions]\label{def:Characterization} 
We consider an auxiliary random variable $W\in \mc{W}$ with $|\mc{W}| = \min\big(|\mc{U}|+1, |\mc{V}|^{|\mc{Z}|}\big)$. The set $\Q_0$ of target probability distributions is defined by
\begin{align}
\Q_0 =& \bigg\{  \PP_{UZ}  \QQ_{W|U} ,  \quad \text{s.t.}, 
\quad\;\;\max_{\PP_X} I( X; Y )  -   I( U ;W |Z )   \geq 0  \bigg\}.\label{eq:SetQ0}
\end{align} 
We define the set $\Q_2\big(\QQ_{UZW}\big)$ of single-letter best-replies of the decoder
\begin{align}
\Q_2\big(\QQ_{UZW}\big) =& \argmin_{\QQ_{V|WZ}}\E_{\QQ_{UZW}\atop  \QQ_{V|WZ}  } \bigg[ d_{\textsf{d}}(U,Z,V) \bigg].\label{eq:SetQ2}
\end{align} 
The encoder's optimal distortion $D_{\textsf{e}}^{\star} $ is given by 
\begin{align}
D_{\textsf{e}}^{\star} =&  \inf_{\QQ_{UZW} \in \Q_0} \max_{\QQ_{V|WZ}  \in  \atop \Q_2(\QQ_{UZW})} \E_{\QQ_{UZW} \atop  \QQ_{V|WZ} } \bigg[d_{\textsf{e}}(U,Z,V)\bigg].\label{eq:PhiOptimalZ}
\end{align}
\end{definition}
We discuss the above definitions.
\begin{itemize}
\item[$\bullet$] The information constraint of the set $\Q_0$ involves the channel capacity $\max_{\PP_X} I( X; Y )$ and the Wyner-Ziv's information rate $ I( U ;W |Z ) = I( U ;W ) - I(W;Z )$, stated in \cite{wyner-it-1976}. It corresponds to the separation result by Shannon \cite{shannon-bell-1948}, extended to the Wyner-Ziv setting by Merhav-Shamai in \cite{MerhavShamai03}. 
\item[$\bullet$] For the clarity of the presentation, the set $\Q_2\big(\QQ_{UZW}\big)$ contains stochastic functions $\QQ_{V|WZ}$, even if for the linear problem \eqref{eq:SetQ2} some optimal $\QQ_{V|WZ}$ are deterministic. If there are several optimal $\QQ_{V|WZ}$, we assume the decoder chooses the one that maximize encoder's distortion: $\max_{\QQ_{V|WZ} \in  \atop \Q_2(\QQ_{UZW})} \E \big[d_{\textsf{e}}(U,Z,V)\big]$.
\item[$\bullet$] The infimum over $\QQ_{UZW} \in \Q_0$ is not a minimum since the  function $\max_{\QQ_{V|WZ} \in  \atop \Q_2(\QQ_{UZW})} \E \big[d_{\textsf{e}}(U,Z,V)\big]$ is not continuous with respect to $\QQ_{UZW}$, see Fig. \ref{fig:EncoderUtility_0.5} and Fig. \ref{fig:EncoderUtility}.  
\item[$\bullet$] In \cite[Theorem IV.2]{LeTreust(ISIT-TwoSided)15}, the sets $\Q_0$ and $\Q_2$  correspond to the target probability distributions $\QQ_{UZW} \QQ_{V|WZ}$ that are achievable for the problem of \emph{empirical coordination}, see also \cite{CuffPermuterCover10}, \cite{LeTreust(EmpiricalCoordination)17}. As noticed in \cite{LeTreustBloch(ISIT)16} and \cite{LeTreustBloch(StateLeakageCoordination)18}, the empirical coordination approach allows us to characterize the ``core of the decoder's knowledge'', which captures what the decoder is able to infer about the random variables involved in the problem.
\item[$\bullet$] The value $D_{\textsf{e}}^{\star} $ corresponds to the Stackelberg equilibrium payoff of an auxiliary one-shot  game in which the decoder chooses $\QQ_{V|WZ}$, knowing in advance that the encoder has chosen $ \QQ_{W|U}\in \Q_0$ and the distortion functions are $\E\big[d_{\textsf{e}}(U,Z,V)\big]$ and $\E\big[d_{\textsf{d}}(U,Z,V)\big]$.
\end{itemize}

\begin{remark}[Equal distortion functions]
When the encoder and the decoder have equal distortion functions $d_{\textsf{d}} = d_{\textsf{e}}$, the set $\Q_2\big(\QQ_{UZW}\big)$ is equal to $ \argmin_{\QQ_{V|WZ}}\E\big[ d_{\textsf{e}}(U,Z,V) \big]$. Thus, we have
\begin{align}
 \max_{\QQ_{V|WZ} \in  \atop \Q_2(\QQ_{UZW})} \E_{\QQ_{UZW} \atop  \QQ_{V|WZ}} \bigg[d_{\textsf{e}}(U,Z,V)\bigg] 
 =&  \min_{\QQ_{V|WZ}} \E_{\QQ_{UZW} \atop  \QQ_{V|WZ}} \bigg[d_{\textsf{e}}(U,Z,V)\bigg].
\end{align} 
Hence, the encoder's optimal distortion $D_{\textsf{e}}^{\star} $ is equal to:
\begin{align}
D_{\textsf{e}}^{\star} =&  \inf_{\QQ_{UZW} \in \Q_0} \max_{\QQ_{V|WZ} \in  \atop \Q_2(\QQ_{UZW} )} \E_{\QQ_{UZW}\atop  \QQ_{V|WZ}} \bigg[d_{\textsf{e}}(U,Z,V)\bigg]\\
=&  \inf_{\QQ_{UZW} \in \Q_0} \min_{\QQ_{V|WZ}} \E_{\QQ_{UZW} \atop  \QQ_{V|WZ}} \bigg[d_{\textsf{e}}(U,Z,V)\bigg]\label{eq:PhiOptimalZr0}\\
=&  \min_{\QQ_{UZW} \in \Q_0,\atop \QQ_{V|WZ}}\E_{\QQ_{UZW} \atop  \QQ_{V|WZ}} \bigg[d_{\textsf{e}}(U,Z,V)\bigg].\label{eq:PhiOptimalZr}
\end{align}
The infimum in \eqref{eq:PhiOptimalZr0} is replaced by a minimum in \eqref{eq:PhiOptimalZr} due to the compactness of $\Q_0$ and the continuity of  $\min_{\QQ_{V|WZ}} \E\Big[d_{\textsf{e}}(U,Z,V)\Big]$ with respect to $\QQ_{UZW}$. We recover the \emph{distortion-rate} function corresponding to \cite[Theorem 1]{MerhavShamai03}.
\end{remark}



\subsection{Main result}\label{sec:MainResult}

We denote $\N^{\star}=\N\setminus\{0\}$ and we characterize the limit of \eqref{eq:GameProblem}.

 \begin{theorem}[Main result]\label{theo:MaxMinStackelberg}
The encoder's long-run distortion satisfies:
\begin{align}
&\forall \varepsilon>0,   \; \exists \bar{n}\in \N^{\star}, \;\forall n\geq \bar{n}, \qquad \inf_{\sigma} \max_{\tau \in \textsf{BR}_{\textsf{d}}(\sigma)} d_{\textsf{e}}^{\,n}(\sigma, \tau)  \leq  D_{\textsf{e}}^{\star}  +  \varepsilon,\label{eq:Achievability}\\
&\forall n \in \N^{\star}, \qquad \inf_{\sigma}\max_{\tau \in \textsf{BR}_{\textsf{d}}(\sigma)} d_{\textsf{e}}^{\,n}(\sigma, \tau)  \geq  D_{\textsf{e}}^{\star}.\label{eq:Converse}
\end{align}
\end{theorem}
When removing the decoder's side information $\mc{Z} = \emptyset$  and changing the infimum to a supremum, we recover the previous result of \cite[Theorem 3.1]{LeTreustTomala19}. The sequence defined by \eqref{eq:GameProblem} is sub-additive. Indeed, when $\sigma$ is the concatenation of several encoding strategies, the concatenation of the corresponding optimal decoding strategies belongs to $ \textsf{BR}_{\textsf{d}}(\sigma)$. Theorem \ref{theo:MaxMinStackelberg} and Fekete's lemma in \cite{Fekete1923}, show that the long-run encoder's distortion converges to its infimum.
\begin{align}
D_{\textsf{e}}^{\star} =& \lim_{n \to +\infty} \inf_{\sigma}\max_{\tau \in \textsf{BR}_{\textsf{d}}(\sigma)} d_{\textsf{e}}^{\,n}(\sigma, \tau)=\inf_{n  \in \N^{\star}} \;\inf_{\sigma}\max_{\tau \in \textsf{BR}_{\textsf{d}}(\sigma)} d_{\textsf{e}}^{\,n}(\sigma, \tau) .\label{eq:limit}
\end{align}

%
%
%
%
%
%
%
%


\subsection{Sketch of proof of Theorem \ref{theo:MaxMinStackelberg}}\label{sec:SketchProof}

We provide some intuitions for the main arguments of the proofs, which are given in App. \ref{sec:AchievabilityProof} and \ref{sec:ConverseProof}.

\textit{Proof of the achievability result \eqref{eq:Achievability}.} We analyse the Bayesian posterior beliefs induced by the concatenation of  Wyner-Ziv's source encoding \cite{wyner-it-1976} and Shannon's channel encoding \cite{shannon-bell-1948}.  We assume that the channel capacity is strictly positive and we consider a probability distribution $\QQ_{UZW}$ such that 1) the information constraints are satisfied with strict inequalities, 2) the set of best-reply symbol $\mc{V}^{\star}\big(z,\QQ_U(\cdot|z,w)\big)$ of Definition \ref{def:BestReply}, is a singleton for all $(z,w)\in \mc{Z}\times\mc{W}$. We introduce Wyner-Ziv's coding rates $\textsf{R}$, $\textsf{R}_{\textsf{L}}$ for the messages $(M,L)$ and we denote by $\eta>0$ the parameter such that 
\begin{eqnarray}
\textsf{R}  + \textsf{R}_{\textsf{L}}& =&       I( U;W )  + \eta  \label{eq:AchievabilityBB1} , \\
\textsf{R}_{\textsf{L}}  &\leq &       I( Z;W )  - \eta  \label{eq:AchievabilityBB2} , \\
\textsf{R} \; & \leq&   \max_{\PP_X} I( X; Y )  -  \eta  \label{eq:AchievabilityBB3}  .
\end{eqnarray}
We introduce the binary random variable $E_{\delta} \in\{0,1\}$ for which $E_{\delta} = 0$  when the messages $(M,L)$ are recovered by the decoder and the sequences $\big(U^n ,  Z^n, W^n, X^n,Y^n \big) $  are jointly typical with tolerance $\delta>0$. 

The major step is to show that the posterior beliefs $ \PP_{\sigma,U_t}(\cdot|y^n,z^n,E_{\delta}=0)$ induced by the encoding strategy $\sigma$ regarding $U_t$ at stage $t\in\{1,\ldots,n\}$, are close on average to the target conditional probability distribution $\QQ_{U_t}(\cdot|w_t,z_t)$:
\begin{align}
&  \E_{\sigma} \Bigg[ \frac{1}{n}  \sum_{t=1}^n D\bigg(  \PP_{\sigma,U_t}(\cdot|Y^n,Z^n,E_{\delta}=0) \bigg| \bigg|   \QQ_{U_t}(\cdot|W_t,Z_t) \bigg)\Bigg] \nonumber \\
\leq&2\delta + \eta + \frac{2}{n}   + 2 \log_2 |\mc{U}| \cdot \PP_{\sigma}\big(E_{\delta}=1 \big) := \epsilon. \label{eq:ControlFinalCF} 
\end{align}
This is the purpose of the proof of Proposition \ref{prop:WynerZivCoding}, stated in Appendix \ref{sec:ProofPropositionCode}.\\

\textit{Proof of the converse result \eqref{eq:Converse}.} For any encoding strategy $\sigma$ of length $n\in\N^{\star}$, we introduce an auxiliary random variable $W=(Y^n,Z^{-T},T)$, where $T$ is the uniform random variable over $\{1,\ldots,n\}$ and  $Z^{-T}$ stands for $(Z_1,\ldots,Z_{T-1},Z_{T+1},\ldots Z_n)$, where $Z_T$ has been removed. We identify $(U,Z)=(U_T,Z_T)$ and we show that the Markov chain $W -\!\!\!\!\minuso\!\!\!\!- U -\!\!\!\!\minuso\!\!\!\!- Z$ is satisfied and that the probability distribution $\PP_{UZW}$ satisfy
\begin{align}
\PP(u,z,w) =&  \frac{1}{n} \cdot  \PP_{\sigma}\big(u_t, z_t,y^n,z^{-t}\big) ,\quad \forall (u,w,z,u^n,z^n,y^n). \label{eq:distributionWW}
\end{align}  
We define $\tilde{\tau}_{V|WZ} = \tau_{V_T|Y^nZ^n}$ and we prove that for both encoder and decoder, the long-run distortion writes
\begin{align}
d^{\,n}_{\textsf{e}}(\sigma,\tau)
=& \sum_{u,z,w} \PP(u,z,w )  \sum_{v}\tilde{\tau}(v| w,z ) \cdot     d_{\textsf{e}}(u,z,v),\label{eq:Reformulation55}
\end{align}
hence 
\begin{align}
 \tau \in \argmin_{\tau'_{V^n|Y^nZ^n}} \E_{\sigma,\tau'}\Bigg[ \frac{1}{n} \sum_{t=1}^n d_{\textsf{d}}(u_t,z_t,v_t)\Bigg]
\Longleftrightarrow&\tilde{\tau}_{V|WZ} \in \Q_2\big(\PP_{UZW}\big).\label{eq:IdentificationN}
\end{align}  
Well known arguments from \cite{wyner-it-1976} and \cite{MerhavShamai03} show that 
\begin{align}
0 \leq&\max_{\PP_X} I(X ; Y)  -  I(U ; W|Z). \label{eq:ConverseW99} 
\end{align}
Therefore, for any encoding strategy $\sigma$ and all $n$, we have
\begin{align}
\max_{\tau \in \textsf{BR}_{\textsf{d}}(\sigma)}  d_{\textsf{e}}^{\,n}(\sigma,\tau) 
\geq&\inf_{ \QQ_{UZW} \in \Q_0} \max_{\QQ_{V|WZ} \in  \atop \Q_2(\QQ_{UZW})} \E_{\QQ_{UZW} \atop  \QQ_{V|WZ}} \bigg[d_{\textsf{e}}(U,Z,V)\bigg]= D_{\textsf{e}}^{\star}.
\end{align}


\subsection{Convex closure formulation}\label{sec:Concavification}

The convex closure of a function $f$ is the largest convex function $\vex f : \mc{X} \rightarrow \R \cup \{- \infty\}$ everywhere smaller than $f$ on $X$. In this section, we reformulate the encoder's optimal distortion $D_{\textsf{e}}^{\star}$ in terms of a  convex closure, similarly to \cite[Corollary 1]{KamenicaGentzkow11} and \cite[Definition 2.4]{LeTreustTomala19}. This alternative approach may simplify the optimization problem in \eqref{eq:PhiOptimalZ}, by plugging the decoder's posterior beliefs and best-reply symbols into the encoder's distortion function. The goal of the strategic communication is to \emph{control the posterior beliefs} of the decoder, knowing it will choose a best-reply symbol afterwards.

Before the transmission, the decoder holds a prior belief corresponding to the source's statistics $\PP_U\in\Delta(\mc{U})$. After observing the pair of symbols $(w,z)\in\mc{W} \times \mc{Z}$, the decoder updates its posterior belief $\QQ_{U}(\cdot|z,w)\in \Delta(\mc{U})$ according to Bayes rule $\QQ(u|z,w) = \frac{\PP(u,z)\QQ(w|u)}{\sum_{u'}\PP(u',z)\QQ(w|u')}$, for all $(u,z,w)\in\mc{U} \times\mc{W} \times \mc{Z}$.

\begin{lemma}[Markov property on posterior belief]\label{lemma:MarkovPropertyPosteriorBelief}
The Markov chain property $Z  -\!\!\!\!\minuso\!\!\!\!-U    -\!\!\!\!\minuso\!\!\!\!-  W$ implies that the \emph{posterior beliefs} $\QQ_{U}(\cdot|z,w)\in \Delta(\mc{U})$ can be expressed from the \emph{interim beliefs} $\QQ_{U}(\cdot|w)\in \Delta(\mc{U})$ since
\begin{align}
\QQ(u|w,z)  =&  \frac{\QQ(u,z,w)  }{\QQ(z,w)  } = \frac{\QQ(u,z|w)  }{\sum_{u'}\QQ(u',z|w)  }= \frac{\QQ(u|w) \PP(z|u)  }{\sum_{u'}\QQ(u'|w) \PP(z|u') }, \qquad \forall (u,z,w)\in \mc{U}\times \mc{Z}\times \mc{W} .\label{eq:ReformulationPosteriors}
\end{align}
\end{lemma}

\begin{definition}[Best-reply]\label{def:BestReply}
For each symbol $z\in \mc{Z}$ and belief $p\in\Delta(\mc{U})$, the decoder chooses the best-reply $v^{\star}(z,p)$ that belongs to the set $\mc{V}^{\star}(z,p)$, defined by
\begin{align}
\mc{V}^{\star}(z,p) = & \argmax_{v \in \argmin \E_{p}\big[d_{\textsf{d}}(U,z,v)\big] }\E_{p}\bigg[d_{\textsf{e}}(U,z,v)\bigg].\label{eq:Vstar}
\end{align}
\end{definition}
If several symbols in $\mc{V}$ are best-reply to $z\in \mc{Z}$ and $p\in\Delta(\mc{U})$, the decoder chooses the worst one for encoder's distortion. This is a reformulation of the maximum in \eqref{eq:PhiOptimalZ}.

\begin{definition}[Robust distortion]\label{def:RobustUtility}
For a symbol $z\in \mc{Z}$ and a belief $p\in\Delta(\mc{U})$, the encoder's \emph{robust distortion function} is defined by
\begin{align}
\psi_{\textsf{e}}(z,p) &=   \E_p\Big[ d_{\textsf{e}}\big(U,z,v^{\star}(z,p)\big)\Big],\label{eq:functionPsi}
\end{align}
where the best-reply $v^{\star}(z,p)$ belongs to the set defined by \eqref{eq:Vstar}.
\end{definition}

\begin{definition}[Average distortion and average entropy]\label{def:AverageUtility}
For each belief $p\in\Delta(\mc{U})$, we define the average distortion function $\Psi_{\textsf{e}}(p)$ and average entropy function $h(p)$ by
\begin{align}
\Psi_{\textsf{e}}(p) =& \sum_{u,z} p(u) \cdot \PP(z|u)\cdot  \psi_{\textsf{e}}\big(z,p_U(\cdot|z)\big),\quad \text{ where } \quad p_U(u|z) = \frac{p(u) \cdot  \PP(z|u)}{\sum_{u'}p(u') \cdot  \PP(z|u')}\; \forall (u,z),\label{eq:FunctionPsip}\\
 h(p) =&   H(p) + \sum_{u} p(u) \cdot H\Big(\PP_Z(\cdot|u)\Big) - H\Big( \sum_{u} p(u) \cdot\PP_Z(\cdot|u)\Big).\label{eq:FunctionHp}
\end{align}
The conditional probability distribution $\PP_Z(\cdot|u)$ is given by the information source. Note that $h(\PP_U) = H(U|Z)$.
\end{definition}


\begin{lemma}[Concavity]\label{lemma:Concavity}
The average entropy $h(p)$ is concave in $p \in \Delta(\mc{U})$.
\end{lemma}
\begin{proof}[Lemma \ref{lemma:Concavity}]
The average entropy $h(p)$ in \eqref{eq:FunctionHp} is equal to the conditional entropy $H(U|Z)$ evaluated with respect to the probability distribution $p\cdot  \PP_{Z|U}\in\Delta(\mc{U}\times \mc{Z})$. The mutual information $I(U;Z)$ is convex in $p\in \Delta(\mc{U})$ (see \cite[pp. 23]{ElGammalKim(book)11}),  and the entropy $H(U)$ is concave in  $p\in \Delta(\mc{U})$. Hence the conditional entropy $h(p) = H(U|Z)= H(U) - I(U;Z)$ is concave in $p\in \Delta(\mc{U})$.
\end{proof}

\begin{theorem}[Convex closure]\label{theo:Concavification}
The solution $D_{\textsf{e}}^{\star} $ of \eqref{eq:PhiOptimalZ} is the convex closure of $ \Psi_{\textsf{e}}(p)$ evaluated at the prior distribution $\PP_U$, under an information constraint,
\begin{align}
D_{\textsf{e}}^{\star}  =&  \inf\bigg\{ \sum_{w} \lambda_w \cdot \Psi_{\textsf{e}}(p_w)   \quad \text{ s.t. } \quad \sum_{w} \lambda_w \cdot p_w  = \PP_U \in \Delta(\mc{U}),\nonumber\\
&\qquad\qquad\qquad\qquad\qquad \text{ and }\quad \sum_{w} \lambda_w \cdot h(p_w)  \geq H(U|Z)  - \max_{\PP_X}I(X;Y) \bigg\},\label{eq:SplittingFormulation}
\end{align}
where the infimum is taken over $\lambda_w\in [0,1]$ summing up to 1 and $p_w\in \Delta(\mc{U})$, for each $w\in \mc{W}$ with $|\mc{W}| = \min\big(|\mc{U}|+1, |\mc{V}|^{|\mc{Z}|}\big)$.
\end{theorem}
The proof of Theorem \ref{theo:Concavification}, stated in App. \ref{sec:ProofConcavification}, rely on the Markov chain property $Z  -\!\!\!\!\minuso\!\!\!\!-U    -\!\!\!\!\minuso\!\!\!\!-  W$. When removing, the decoder's side information, e.g. $|\mc{Z}|=1$, and changing the infimum into a supremum, we recover the value of the optimal splitting problem of \cite[Definition 2.4]{LeTreustTomala19}. The ``splitting Lemma'' by Aumann and Maschler \cite{AM95}, also called ``Bayes plausibility'' in \cite{KamenicaGentzkow11}, ensures that the strategy $\QQ(w|u) =  \frac{\lambda_w \cdot p_w(u)}{\PP(u)}$ induces the collection of posterior beliefs $(\lambda_w,p_w)_{w\in \mc{W}}$, also referred to as ``the splitting of the prior belief''. Formulation \eqref{eq:SplittingFormulation} provides an alternative point of view on the encoder's optimal distortion \eqref{eq:PhiOptimalZ}.
\begin{itemize}
\item[$\bullet$] The optimal solution $D_{\textsf{e}}^{\star} $ can be found by adapting the concavification method \cite{AM95}, to the minimization problem. In Sec \ref{sec:ExampleBinary}, we compute explicitly the optimal strategy for the Wyner-Ziv's example of the \textit{doubly symmetric binary source} (DSBS), in  \cite[Sec. II]{wyner-it-1976}.
\item[$\bullet$] When the channel is perfect and has a large input alphabet $|\mc{X}|\geq \min(|\mc{U}|,|\mc{V}|^{|\mc{Z}|})$, the optimal solution is obtained by removing the information constraint \eqref{eq:SplittingFormulation}. This noise-free setting is related to the problem of persuasion with heterogeneous beliefs,  investigated in \cite{AlonsoCamara(JET)2016} and \cite{LaclauRenou17}.
\item[$\bullet$] The information constraint $\sum_{w} \lambda_w \cdot h(p_w)  \geq H(U|Z)  - \max_{\PP_X}I(X;Y)$ in \eqref{eq:SplittingFormulation} is a reformulation of $I(U;W|Z) \leq  \max_{\PP_X}I(X;Y)$ in \eqref{eq:SetQ0}, since 
\begin{align}
\sum_{w} \lambda_w \cdot h(p_w) =& \sum_{w} \lambda_w \cdot H(U|Z,W=w) =  H(U|Z,W).
\end{align} 
\item[$\bullet$]The dimension of the problem \eqref{eq:SplittingFormulation}  is $|\mc{U}|$. Caratheodory's Lemma (see \cite[Corollary 17.1.5, pp. 157]{rockafellar1970convex} and \cite[Corollary A.2, pp. 26]{LeTreustTomala19}) provides the cardinality bound $|\mc{W}| = |\mc{U}|+1$. 
\item[$\bullet$] The cardinality of $\mc{W}$ is also restricted by the vector of recommended symbols $|\mc{W}|  = |\mc{V}|^{|\mc{Z}|}$, telling to the decoder which symbol $v\in\mc{V}$ to return when the side information is $z\in\mc{Z}$. Otherwise assume that two posteriors $p_{w_1}$ and $p_{w_2}$  induce the same vectors of symbols $v^1=(v_1^1,\ldots,v_{|\mc{Z}|}^1) = v^2=(v_1^2,\ldots,v_{|\mc{Z}|}^2)$. Then, both posteriors $p_{w_1}$ and $p_{w_2}$ can be replaced by their average: 
\begin{align}
\widetilde{p}  = \frac{\lambda_{w_1} \cdot p _{w_1} + \lambda_{w_2}\cdot p _{w_2}}{\lambda_{w_1} + \lambda_{w_2}},
\end{align}
without changing the distortion and still satisfying the information constraint:
\begin{align}
& h(\widetilde{p})  \geq \frac{\lambda_{w_1} \cdot h(p _{w_1}) + \lambda_{w_2}\cdot h(p _{w_2})}{\lambda_{w_1} + \lambda_{w_2}}\label{eq:EntropyTilde}\\
\Longrightarrow& \sum_{w \neq w_1, \atop w \neq w_2} \lambda_{w} \cdot h(p_w) + (\lambda_{w_1} + \lambda_{w_2})\cdot h(\widetilde{p}) \geq H(U|Z) - \max_{\PP_X}I(X;Y).\label{eq:telling}
\end{align}
Inequality \eqref{eq:EntropyTilde} comes from the concavity of $h(p)$, stated in Lemma \ref{lemma:Concavity}.
\end{itemize}

The splitting under information constraint of Theorem \ref{theo:Concavification} can be reformulated in terms of Lagrangian and in terms of the convex closure of $\tilde{\Psi}_{\textsf{e}}(p,\nu)$ defined by 
\begin{align}
\tilde{\Psi}_{\textsf{e}}(p,\nu) &= 
\begin{cases}
\Psi_{\textsf{e}}(p), \text{ if } \nu \leq h(p),\\
+ \infty, \text{ otherwise. }  
\end{cases}
\end{align}

\begin{theorem}\label{theo:Concavification2}
The optimal solution $D_{\textsf{e}}^{\star}$ reformulates as:
\begin{align}
D_{\textsf{e}}^{\star}  =& \sup_{t\geq 0 } \bigg\{ \vex \Big[\Psi_{\textsf{e}} + t \cdot h\Big]\big(\PP_U\big) - t \cdot \Big(H(U|Z) - \max_{\PP_X}I(X;Y)\Big)\bigg\}\label{eq:Lagrangian}\\
=&\vex \tilde{\Psi}_{\textsf{e}}\Big(\PP_U,H(U|Z) - \max_{\PP_X}I(X;Y)\Big).\label{eq:GeneralCav}
\end{align}
\end{theorem}
Equation \eqref{eq:Lagrangian} is the convex closure of a Lagrangian that integrates the information constraint and equation \eqref{eq:GeneralCav} corresponds to the convex closure of a a bi-variate function where the information constraint requires an additional dimension. The proof follows directly from \cite[Theorem 3.3, pp. 37]{LeTreustTomala19}, by replacing concave closure by convex closure.

%


\section{Example with binary source and side information}\label{sec:ExampleBinary}

We consider a binary source $U \in \{u_0,u_1\}$ with probability distribution $\PP(u_1)=p_0 \in[0,1]$. The binary side information $Z\in \{z_0,z_1\}$ is drawn according to the conditional probability distribution $\PP(z|u)$ with parameter $\delta_0\in [0,1]$ and $\delta_1\in [0,1]$. In this section, we consider distortion functions $d_{\textsf{e}}$ and $d_{\textsf{d}}$ that do not depend on the side information $Z$. The cardinality bound in Definition \ref{def:Characterization} is $ |\mc{W}| =\min\big(|\mc{U}|+1, |\mc{V}|^{|\mc{Z}|}\big) =3$, hence the random variable $W$ is drawn according to the conditional probability distribution $\QQ(w|u)$ with parameters $(\alpha_k,\beta_k)_{k\in \{1,2,3\}}\in [0,1]^6$ such that $\sum_k \alpha_k = \sum_k \beta_k = 1$. The joint probability distribution $\PP(u,z) \QQ(w|u)$ is depicted in Fig. \ref{fig:SignalingZ}.
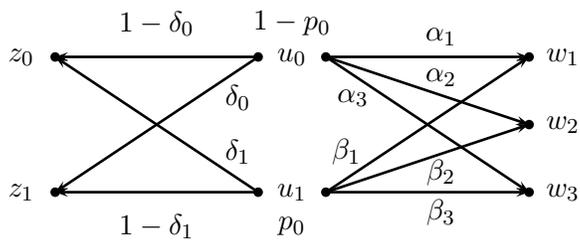
\begin{figure}[!ht]
\begin{center}
\psset{xunit=0.9cm,yunit=0.9cm}
\begin{pspicture}(-2.5,0)(3,3)
\rput(0,0.5){$u_1$}
\rput(0,2.5){$u_0$}
\psdots(0.5, 0.5)(0.5, 2.5)(-0.5, 0.5)(-0.5, 2.5)(3.5,0.5)(3.5,2.5)(-3.5,0.5)(-3.5,2.5)(3.5,1.5)
\psline[linewidth=1pt]{<-}(-3.5,2.5)(-0.5,2.5)
\psline[linewidth=1pt]{<-}(-3.5,0.5)(-0.5,0.5)
\psline[linewidth=1pt]{<-}(-3.5,0.5)(-0.5,2.5)
\psline[linewidth=1pt]{<-}(-3.5,2.5)(-0.5,0.5)
\psline[linewidth=1pt]{->}(0.5,0.5)(3.5,0.5)
\psline[linewidth=1pt]{->}(0.5,2.5)(3.5,2.5)
\psline[linewidth=1pt]{->}(0.5,0.5)(3.5,1.5)
\psline[linewidth=1pt]{->}(0.5,2.5)(3.5,1.5)
\psline[linewidth=1pt]{->}(0.5,0.5)(3.5,2.5)
\psline[linewidth=1pt]{->}(0.5,2.5)(3.5,0.5)
\rput(4,0.5){$w_3$}
\rput(4,2.5){$w_1$}
\rput(4,1.5){$w_2$}
\rput(-4,0.5){$z_1$}
\rput(-4,2.5){$z_0$}
\rput(2.2,2.8){$\alpha_1$}
\rput(2.2,0.2){$\beta_3$}
\rput(0,3){$1-p_0$}
\rput(0,0){$p_0$}
\rput(0.9,1.9){$ \alpha_3$}
\rput(2.2,2.2){$ \alpha_2$}
\rput(0.8,1.1){$ \beta_1$}
\rput(2.2,0.8){$ \beta_2$}
\rput(-2,3){$1 - \delta_0$}
\rput(-0.8,1.9){$ \delta_0$}
\rput(-2,0){$1 -  \delta_1$}
\rput(-0.8,1.1){$ \delta_1$}
\end{pspicture}
\caption{Joint probability distribution $\PP(u,z) \QQ(w|u)$ depending on parameters $p_0$, $\delta_0$, $\delta_1$,  $(\alpha_k,\beta_k)_{k\in \{1,2,3\}}$.}
\label{fig:SignalingZ}
\end{center}
\end{figure}

The encoder minimizes the Hamming distortion $d_{\textsf{e}}(u,v)$ given by Fig. \ref{fig:Utility1}. The decoder's distortion in Fig. \ref{fig:Utility2} includes an extra cost $\kappa\in[0,1]$ when it returns the symbol $v_1$ instead of the symbol $v_0$. The extra cost $\kappa$ may capture a computing cost, an energy cost, or the fact that an estimation error for the symbol $v_1$ is more harmful than an estimation error for the symbol $v_0$.
\begin{figure}[ht!]
\begin{minipage}{0.495\linewidth}
\begin{center}
\psset{xunit=1cm,yunit=1cm}
\begin{pspicture}(0,0.2)(2,2.2)
\psframe(0,0)(2,2)
\psline(1,0)(1,2)
\psline(0,1)(2,1)
\rput(-0.3,0.5){$u_1$}
\rput(-0.3,1.5){$u_0$}
\rput(0.5,2.3){$v_0$}
\rput(1.5,2.3){$v_1$}
\rput(0.5,1.5){$0$}
\rput(0.5,0.5){$1$}
\rput(1.5,1.5){$1$}
\rput(1.5,0.5){$0$}
\end{pspicture}
\caption{Encoder's distortion function $d_{\textsf{e}}(u,v)$.}\label{fig:Utility1}
\end{center}
\end{minipage}\hfill
\begin{minipage}{0.495\linewidth}
\begin{center}
\psset{xunit=1cm,yunit=1cm}
\begin{pspicture}(0,0.2)(2,2.2)
\psframe(0,0)(2,2)
\psline(1,0)(1,2)
\psline(0,1)(2,1)
\rput(-0.3,0.5){$u_1$}
\rput(-0.3,1.5){$u_0$}
\rput(0.5,2.3){$v_0$}
\rput(1.5,2.3){$v_1$}
\rput(0.5,1.5){$0$}
\rput(0.5,0.5){$1$}
\rput(1.5,1.5){$1+\kappa$}
\rput(1.5,0.5){$\kappa$}
\end{pspicture}
\caption{Decoder's distortion $d_{\textsf{d}}(u,v)$ with extra cost $\kappa\in [0,1]$.}\label{fig:Utility2}
\end{center}
\end{minipage}\hfill
\end{figure}

\subsection{Decoder's best-reply}
After receiving the pair of symbols $(w,z)$, the decoder updates its \emph{posterior belief} $\QQ_U(\cdot|w,z)\in \Delta(\mc{U})$, according to Bayes rule. We denote by $p =\QQ(u_1|w,z)\in[0,1]$ the parameter of the posterior belief. Given the extra cost is $\kappa = \frac34$ and we denote by $\gamma = \frac{1 + \kappa}{2}=\frac78$ the \emph{belief threshold} at which the decoder changes its symbol, as in Fig. \ref{fig:BestReply}. When the decoder's belief is exactly equal to the threshold $p = \gamma=\frac78$, the decoder is indifferent between the two symbols $\{v_0,v_1\}$, by convention we assume that it chooses $v_0$, i.e. the worst symbol for the encoder.  Hence the decoder chooses a best-reply $v_0^{\star}$ or $v_1^{\star}$ depending on the interval $[0,\gamma]$ or $(\gamma,1]$ in which lies the belief parameter $p\in[0,1]$, see Fig. \ref{fig:BestReply}.

\begin{figure}[!ht]
\begin{center}
\psset{xunit=0.5cm,yunit=0.2cm}
\begin{pspicture}(-1,-6)(10,22)
\psline[linewidth=1pt]{->}(0,-1)(0,21)
\psline[linewidth=1pt](10,-1)(10,21)
\psline[linewidth=1pt]{->}(-0.5,0)(12,0)
\rput(-0.5,-0.8){0}
\rput(10.5,-0.8){1}
\rput(12,-1.5){$p$}
\psline[linecolor=vert, linewidth=3.5pt](0,17.5)(10,7.5)
\rput[r](-0.1,22){$\E_p\big[d_{\textsf{d}}(U,v)\big]$}
\rput[r](-0.4,17.5){$1 + \kappa = \frac74$}
\rput[l](10.2,7.5){$\kappa = \frac34$}
\rput[l](10.2,11){1}
\psline[linecolor=blue, linewidth=3.5pt](0,0)(10,10)
\rput[u](1,18.5){$\textcolor[rgb]{0,0.6,0}{v_1}$}
\rput(1,2.7){$\textcolor[rgb]{0,0,1}{v_0}$}
\psline[linecolor=red, linewidth=1pt](0,0)(8.75,8.75)(10,7.5)
\psline[linestyle=dotted,linecolor=red](8.75,0)(8.75,8.75)
\rput[l]{-45}(8.75,-1){$\textcolor[rgb]{1,0,0}{\gamma = \frac{1 + \kappa}{2} = \frac78}$}
\psdots[linecolor = red,dotscale=1.5](8.75 ,0)
\end{pspicture}
\caption{Decoder's expected distortion $\E_p\big[d_{\textsf{d}}(U,v)\big] = (1-p)\cdot d_{\textsf{d}}(u_0,v) +  p \cdot d_{\textsf{d}}(u_1,v)$ for $v\in\{v_0,v_1\}$, depending on the belief parameter $p = \QQ(u_1|w,z)\in [0,1]$. For an extra cost $\kappa = \frac34$, the decoder's best-reply is the symbol  $\textcolor[rgb]{0,0,1}{v_0^{\star}}$ if the posterior belief $p\in[0,\gamma]$ and $\textcolor[rgb]{0,0.6,0}{v_1^{\star}}$ if $p\in(\gamma,1]$ with $\gamma = \frac78$.}\label{fig:BestReply}
\end{center}
\end{figure}

\subsection{Interim and posterior belief}
The correlation between random variables $(U,Z)$ is fixed whereas the correlation between random variables  $(U,W)$ is  selected by the encoder. This imposes a strong relationship between the \emph{iterim belief} $\QQ_{U|W}$ induced by the encoder, and the \emph{posterior belief} $\QQ_{U|WZ}$ that determine the decoder's best-reply symbol $v$. For a symbol $w\in\mc{W}$, we denote the \textit{interim belief} by $q=\QQ(u_1|w) \in [0,1]$. Lemma \ref{lemma:MarkovPropertyPosteriorBelief} ensures that the posterior belief depending on the side information $z_0$ or $z_1$, are given by
\begin{align}
\QQ(u_1|w,z_0) =& \frac{\QQ(u_1|w) \cdot \PP(z_0|u_1)  }{\QQ(u_0|w) \cdot \PP(z_0|u_0)  + \QQ(u_1|w)\cdot  \PP(z_0|u_1) } =  \frac{q \cdot \delta_1  }{(1-q) \cdot (1-\delta_0) + q \cdot \delta_1 }=:p_0(q)\label{eq:VeritableBelief0},\\
\QQ(u_1|w,z_1) =& \frac{\QQ(u_1|w) \cdot \PP(z_1|u_1)  }{\QQ(u_0|w)\cdot \PP(z_1|u_0)  + \QQ(u_1|w) \cdot\PP(z_1|u_1) } =  \frac{q \cdot (1-\delta_1 ) }{(1-q) \cdot \delta_0 + q \cdot  (1-\delta_1 ) }=:p_1(q)\label{eq:VeritableBelief1}.
\end{align}
The posterior beliefs after receiving the side information $z_0$ or $z_1$ are related to the \textit{interim belief} $q\in[0,1]$ through the two functions $p_0(q)$, $p_1(q)$, depicted on Fig. \ref{fig:DynamicNew0.5_0.05_0.5_0.875}. 
\FigDynamicNew{0.5}{0.05}{0.5}{0.875}

\subsection{Encoder's average distortion function}

Given the \emph{belief threshold} $\gamma=\frac78$ at which the decoder changes its symbol, we define the parameters $\nu_0$ and $\nu_1$ such that $p_0(\nu_0) = \gamma$ and $p_1(\nu_1) = \gamma$.
\begin{align}
\gamma = p_0(\nu_0) \;\; \Longleftrightarrow&\;\;  \nu_0 = \frac{\gamma \cdot (1 - \delta_0)}{\delta_1 \cdot (1 - \gamma) + \gamma \cdot (1 - \delta_0)},\\
\gamma = p_1(\nu_1) \;\; \Longleftrightarrow&\;\;  \nu_1 = \frac{\gamma \cdot   \delta_0}{\gamma \cdot   \delta_0  + (1 - \delta_1) \cdot (1 - \gamma) }.
\end{align}
\begin{remark}
We have the equivalence $\nu_1<\nu_0 \;\; \Longleftrightarrow  \;\; \delta_0 + \delta_1 < 1$.
\end{remark}
Since the distortion functions given in Fig. \ref{fig:Utility1} and Fig. \ref{fig:Utility2} do not depend on the side information $z$, we denote by $\psi_{\textsf{e}}(p)$ the \textit{robust distortion function} of Definition \ref{def:RobustUtility}, given by
\begin{align}
\psi_{\textsf{e}}(p) &=  \min_{v \in \argmin \atop (1-p)\cdot d_{\textsf{d}}(u_0,v) +  p \cdot d_{\textsf{d}}(u_1,v)} (1-p) \cdot d_{\textsf{e}}(u_0,v) + p\cdot d_{\textsf{e}}(u_1,v),\\
&=\UN\big(p\leq \gamma \big) \cdot \Big( (1-p) \cdot d_{\textsf{e}}(u_0,v_0) +p\cdot d_{\textsf{e}}(u_1,v_0) \Big) \nonumber \\
&+ \UN\big(p > \gamma \big)\cdot \Big( (1-p) \cdot d_{\textsf{e}}(u_0,v_1) +p\cdot d_{\textsf{e}}(u_1,v_1) \Big) \\
&=p \cdot  \UN\big(p\leq \gamma\big) + (1-p) \cdot  \UN\big(p >\gamma\big).\label{eq:FunctionPsi1}
\end{align}
Without loss of generality, we assume that $\delta_0 + \delta_1 < 1$, hence $\nu_1<\nu_0$. The average distortion function $\Psi_{\textsf{e}}(q)$ of Definition \ref{def:AverageUtility} depends on the interim belief parameter $q\in[0,1]$ as follows
\begin{align}
\Psi_{\textsf{e}}(q) =& \prob_q(z_0) \cdot \psi_{\textsf{e}}\big(p_0(q)\big) +  \prob_q(z_1)\cdot\psi_{\textsf{e}}\big(p_1(q)\big)\\
=& \Big( (1-q)\cdot(1-\delta_0)  + q \cdot \delta_1\Big) \cdot \Big( p_0(q) \cdot  \UN\big(p_0(q) \leq \gamma\big) + (1-p_0(q)) \cdot  \UN\big(p_0(q) >\gamma \big)\Big)\nonumber \\
&+ \Big( (1-q)\cdot \delta_0  +q \cdot (1-\delta_1) \Big)\cdot  \Big( p_1(q) \cdot  \UN\big(p_1(q) \leq \gamma\big) + (1-p_1(q)) \cdot  \UN\big(p_1(q) > \gamma \big)\Big)\\
=&q \cdot  \UN\big(q \leq \nu_1\big) + \big(q \cdot \delta_1 +  (1-q) \cdot \delta_0\big)\cdot  \UN\big(\nu_1<q \leq \nu_0\big) +(1-q) \cdot  \UN\big(q > \nu_0\big) . \label{eq:AverageFunctionPsi1}
\end{align}
In Fig. \ref{fig:DSBS_02}, Fig. \ref{fig:DSBS_04} and Fig. \ref{fig:EncoderUtility}, the average distortion function $\Psi_{\textsf{e}}(q)$ are represented by the orange lines, whereas the black curve is the average entropy $h(q)$ defined by
\begin{align}
h(q)  =& H_b(q)  + (1 - q) \cdot H_b(\delta_0) + q \cdot H_b(\delta_1)- H_b\Big((1 - q) \cdot \delta_0+ q \cdot (1 - \delta_1)\Big).
 \end{align}

\subsection{Optimal splitting with three posteriors}

Since the cardinality bound  is $ |\mc{W}| =\min\big(|\mc{U}|+1, |\mc{V}|^{|\mc{Z}|}\big) =3$, we consider a splitting of the prior $p_0$ in three posteriors $(q_1,q_2,q_3)\in [0,1]^3$ with respective weights $(\lambda_1,\lambda_2,\lambda_3)\in [0,1]^3$, defined by \eqref{eq:SystemEquation1},  \eqref{eq:SystemEquation2}. 
\begin{align}
1 =&  \lambda_1  +  \lambda_2 +  \lambda_3,\label{eq:SystemEquation1}\\
p_0 =& \lambda_1 \cdot q_1 +  \lambda_2 \cdot q_2 +  \lambda_3 \cdot q_3,\label{eq:SystemEquation2}\\
H(U|Z) - C =& \lambda_1 \cdot h(q_1) +  \lambda_2 \cdot h(q_2) +  \lambda_3 \cdot h(q_3),\label{eq:SystemEquation3}
\end{align}
Equation \eqref{eq:SystemEquation3} is satisfied when the information constraint is binding. By inverting the system  \eqref{eq:SystemEquation1}-\eqref{eq:SystemEquation3}, we obtain
\begin{align}
\lambda_1 =& \frac{ \big(H(U|Z) - C\big) \cdot (q_2  -  q_3) + h(q_2)\cdot (q_3 - p_0) +  h(q_3)\cdot (p_0 - q_2 )  }{ h(q_1) \cdot (q_2  -  q_3)   + h(q_2) \cdot  (q_3 - q_1)   + h(q_3) \cdot (q_1 - q_2 ) } , \label{eq:Lambda1}\\
\lambda_2 =& \frac{\big(H(U|Z) - C\big) \cdot (q_3  - q_1)  + h(q_3) \cdot (q_1  - p_0) + h(q_1) \cdot(p_0  - q_3)}{ h(q_1) \cdot (q_2  -  q_3)   + h(q_2) \cdot  (q_3 - q_1)   + h(q_3) \cdot (q_1 - q_2 ) } , \label{eq:Lambda2}\\
\lambda_3 =&  \frac{\big(H(U|Z) - C\big) \cdot (q_1 - q_2) + h(q_1) \cdot (q_2  -  p_0)   + h(q_2) \cdot  (p_0 - q_1) }{ h(q_1) \cdot (q_2  -  q_3)   + h(q_2) \cdot  (q_3 - q_1)   + h(q_3) \cdot (q_1 - q_2 ) } .\label{eq:Lambda3}
\end{align}
The triple of posteriors $(q_1,q_2,q_3)$ is feasible if and only if  the weights $(\lambda_1,\lambda_2,\lambda_3)$ belong to the interval $[0,1]^3$. The average distortion $\Psi_{\textsf{e}}(q)$ is piece-wise linear, hence the optimal triple of posteriors may belong to distinct intervals $q_1\in [0, \nu_1)$, $q_2\in [\nu_1,\nu_2)$, $q_3\in [\nu_2,1]$. Otherwise, take the average of two posteriors of the same interval which provides the same distortion value and has a larger entropy, due to the strict concavity of entropy function. 

Aumann and Maschler's splitting lemma \cite{AM95} or Kamenica and Gentzkow's Bayes plausibility \cite{KamenicaGentzkow11} claim that the splitting $p_0 = \lambda_1 \cdot q_1 +  \lambda_2 \cdot q_2 + \lambda_3 \cdot q_3$  is implemented by the following strategy \begin{align}
\QQ(w_k|u_0) =&  \QQ(w_k) \cdot \frac{1-\QQ(u_1|w_k)}{1-\PP(u_1)} =  \lambda_k \cdot \frac{1-q_k}{1-p_0} =:\alpha_k , \qquad k \in \{,1,2,3\}\\
\QQ(w_k|u_1) =& \QQ(w_k) \cdot  \frac{\QQ(u_1|w_k)}{\PP(u_1)} = \lambda_k \cdot \frac{q_k}{p_0} =:\beta_k , \qquad k \in \{,1,2,3\}.
\end{align}

\begin{figure}[!ht]
\begin{center}
\psset{xunit=5cm,yunit=5cm}
\begin{pspicture}(0,-0.45)(1,1)
\psline{->}(0,-0.1)(0,1.1)
\psline{->}(-0.1,0)(1.1,0)
\psline{-}(1,-0.1)(1,1)
\rput[u](-0.03,-0.03){$0$}
\rput[l](-0.05,1){$1$}
\rput[u](1.05,-0.05){$1$}
\rput[u](1.25,-0.05){$q$}
\psdots(0.5,0)
\psline[linestyle=dashed](0.5,0)(0.5,0.8813)
\rput[l]{-45}(0.5,-0.05){$q_2 =  p_0 = 0.5$}
\psdots(0.3,0)(0.7,0)
\psline[linestyle=dashed](0.3,0)(0.3,0.3)
\psline[linestyle=dashed](0.7,0)(0.7,0.3)
\rput[l]{-45}(0.69,-0.05){$ \nu_1 = 0.7$}
\rput[l]{-45}(0.31,-0.05){$ \nu_2 = 0.3$}
\psline[linestyle=dotted,linecolor=black](0, 0.3) (1,0.3) 
\rput[r](-0.02,0.3){$\delta_0=\delta_1=0.3$}
\psline[linecolor=orange, linewidth=3.5pt](0,0)(0.3, 0.3)
\psline[linecolor=orange, linewidth=3.5pt](1,0)(0.7, 0.3)
\psline[linecolor=orange, linewidth=3.5pt](0.3, 0.3)(0.7, 0.3)
\PlotConditionalEntropy{0.5}{0.3}{0.3}
\psline[linecolor=black](0,0.6813)(1, 0.6813)
\psdots[linecolor=black](0.5, 0.6813)(0,0.6813)(0, 0.3)(1,0.5375)
\rput[r](-0.02,0.6813){$H(U|Z) - C = 0.6813$}
\psline[linestyle=dashed,linecolor=black](0.5, 0.2098)(1, 0.2098)
\psline[linestyle=dashed,linecolor=black](0.8813, 0.5375)(1,0.5375)
\rput[l](1.04,0.5375){$h(q^{\star}) = 0.5375$}

\psdots[linecolor=red,dotscale=1.5](0.145, 0.145)(0.855,0.145) (0.5,0.3)
\psdots[linecolor=red](0.145, 0.5375)(0.5, 0.8813)(0.855, 0.5375)
\psline[linestyle=dotted,linecolor=red](0.145, 0.5375)(0.145, 0)
\psline[linestyle=dotted,linecolor=red](0.855, 0.5375)(0.855, 0)
\psline[linestyle=dotted,linecolor=red](0.145, 0.5375)(0.5, 0.8813)(0.855, 0.5375)(0.145, 0.5375)
\rput[l]{-45}(0.855,-0.05){$q_3=1 -q^{\star} = 0.855$}
\rput[l]{-45}(0.145,-0.05){$q_1 =q^{\star}= 0.145$}
\psdots[linecolor=black](0.145, 0)(0.855, 0)

\psline[linestyle=dashed,linecolor=red](0.145, 0.145)(0.855,0.145) (0.5,0.3)(0.145, 0.145)

\psline[linecolor=red](0.145, 0.5375)(0.5, 0.8813)(0.855, 0.5375)(0.145, 0.5375)
\psdots[linecolor=red](0.5, 0.2098)
\psdots[linecolor=black](1, 0.2098)
\rput[l](1.02,0.2098){$D_{\textsf{e}}^{\star} = 0.2098$}
\end{pspicture}
\caption{Wyner-Ziv's example for DSBS. When $C\in[0,H(U|Z) - h(q^{\star}) ]$ the optimal splitting is $(q_1,q_2,q_3)=(q^{\star},\frac12,1-q^{\star})$. For the parameters $p_0 = 0.5$, $\delta_0=\delta_1=0.3$, $C = 0.2$, $\kappa=0$, the optimal encoder's distortion is $D_{\textsf{e}}^{\star} = 0.2098$.}
\label{fig:DSBS_02}
\end{center}
\end{figure}

\begin{figure}[!ht]
\begin{center}
\psset{xunit=5cm,yunit=5cm}
\begin{pspicture}(0,-0.3)(1,1)
\psline{->}(0,-0.1)(0,1.1)
\psline{->}(-0.1,0)(1.1,0)
\psline{-}(1,-0.1)(1,1)
\rput[u](-0.03,-0.03){$0$}
\rput[l](-0.05,1){$1$}
\rput[u](1.05,-0.05){$1$}
\rput[u](1.25,-0.05){$q$}
\psdots(0.5,0)
\psline[linestyle=dashed](0.5,0)(0.5,0.8813)
\rput[l]{-45}(0.5,-0.05){$q_2 =  p_0 = 0.5$}
\psdots(0.3,0)(0.7,0)
\psline[linestyle=dashed](0.3,0)(0.3,0.3)
\psline[linestyle=dashed](0.7,0)(0.7,0.3)
\rput[l]{-45}(0.69,-0.05){$ \nu_1 = 0.7$}
\rput[l]{-45}(0.31,-0.05){$ \nu_2 = 0.3$}
\psline[linestyle=dotted,linecolor=black](0, 0.3) (1,0.3) 
\rput[r](-0.03,0.3){$\delta_0=\delta_1=0.3$}
\psline[linecolor=orange, linewidth=3.5pt](0,0)(0.3, 0.3)
\psline[linecolor=orange, linewidth=3.5pt](1,0)(0.7, 0.3)
\psline[linecolor=orange, linewidth=3.5pt](0.3, 0.3)(0.7, 0.3)
\PlotConditionalEntropy{0.5}{0.3}{0.3}
\psline[linecolor=black](0,0.4813)(1, 0.4813)
\psdots[linecolor=black](0,0.4813)(0, 0.3)(1,0.5375)(0.5, 0.8813)(0.145, 0.5375)(0.855, 0.5375)
\rput[r](-0.03,0.4813){$H(U|Z) - C = 0.4813$}
\psline[linestyle=dashed,linecolor=black](0.1212, 0.1212)(1, 0.1212)
\psline[linestyle=dashed,linecolor=black](0.855, 0.5375)(1,0.5375)
\psline[linecolor=black](0.145, 0.5375)(0.855, 0.5375)(0.5,0.8813)(0.145, 0.5375)
\rput[l](1.04,0.5375){$h(q^{\star}) = 0.5375$}
\psdots[linecolor=red,dotscale=1.5](0.1212, 0.1212)(0.8788,0.1212) 
\psdots[linecolor=red](0.1211,  0.4813)(0.8788,  0.4813)
\psline[linestyle=dotted,linecolor=red](0.1212, 0.4813)(0.1212, 0)
\psline[linestyle=dotted,linecolor=red](0.8788, 0.4813)(0.8788, 0)
\psline[linestyle=dotted,linecolor=black](0.145, 0.5375)(0.145, 0)
\psline[linestyle=dotted,linecolor=black](0.855, 0.5375)(0.855, 0)
\rput[l]{-45}(0.8,-0.05){$1-q^{\star} = 0.855$}
\rput[l]{-45}(0.2,-0.05){$q^{\star} = 0.145$}
\rput[l]{-45}(0.9,-0.05){$q_3 = 0.8788$}
\rput[l]{-45}(0.1,-0.05){$q_1 = 0.1212$}
\psdots[linecolor=black](0.145, 0)(0.855, 0)(0.1212, 0)(0.8788, 0)
\psline[linecolor=red](0.1212,  0.4813)(0.8788,  0.4813)
\psdots[linecolor=black](1, 0.1212)
\rput[l](1.02,0.1212){$D_{\textsf{e}}^{\star} = 0.1212$}
\end{pspicture}
\caption{Wyner-Ziv's example for DSBS. When $C\in[H(U|Z) - h(q^{\star}), H(U|Z)]$ the optimal splitting has two posteriors $(q_1,q_3)=\Big(h^{-1}\big(H(U|Z)-C\big),1-h^{-1}\big(H(U|Z)-C\big)\Big)$. For parameters $p_0 = 0.5$, $\delta_0=\delta_1=0.3$, $C = 0.4$, $\kappa=0$, the optimal encoder's distortion is $D_{\textsf{e}}^{\star} =  0.1212$.}
\label{fig:DSBS_04}
\end{center}
\end{figure}
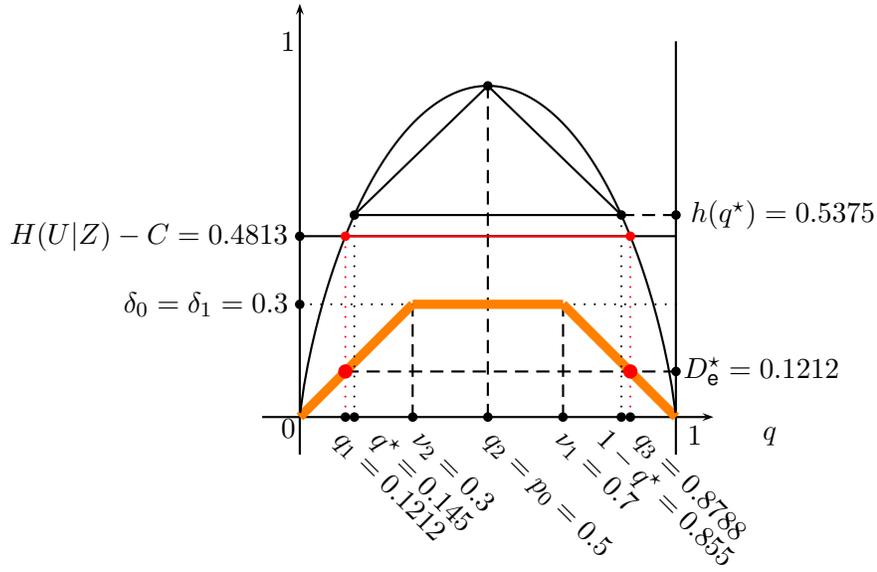

\subsection{Wyner-Ziv's example for DSBS with $p_0 = 0.5$, $\delta_0=\delta_1=0.3$, $\kappa=0$}\label{sec:TwoPosteriorsNoCost}

We investigate the example of \textit{doubly symmetric binary source} (DSBS) with $\delta_0=\delta_1=0.3$ whose solution is characterized in \cite[Sec. II, pp. 3]{wyner-it-1976}. In this example, both encoder and decoder minimize the Hamming distortion, hence $\kappa=0 \Longleftrightarrow \gamma=\frac12$. We introduce the notation $q\star \delta := (1 - q) \cdot \delta+ q \cdot (1 - \delta)$, the average distortion and average entropy write
\begin{align}
\Psi_{\textsf{e}}(q) =& q \cdot  \UN\big(q \leq \delta \big) + \delta\cdot  \UN\big(\delta<q \leq1- \delta\big) +(1-q) \cdot  \UN\big(q > 1-\delta\big),\\
h(q) =&  H(U|Z) + H_b(q)  - H_b\big(q\star \delta\big),
\end{align}
We remark that $H(U|Z) - h(q) = H_b\big(q \star \delta\big) -H_b(q)$. 

\begin{figure}[!ht]
\begin{center}
\psset{xunit=5cm,yunit=5cm}
\begin{pspicture}(0,-0.2)(1,1.1)
\psline{->}(0,-0.1)(0,1.1)
\psline{->}(-0.1,0)(1.1,0)
\psline{-}(1,-0.1)(1,1)
\rput[u](-0.03,-0.03){$0$}
\rput[r](-0.04,1){$C$}
\rput[u](1.05,-0.05){$1$}
\rput[u](1.25,-0.05){$D_{\textsf{e}}^{\star}$}
\psdots(0.5,0)
\rput[l]{-45}(0.5,-0.05){$p_0 = 0.5$}
\PlotGq{0.5}{0.3}{0.3}
\psline[linestyle=dashed,linecolor=black](0, 0.3438)(0.145, 0.3438)(0.145, 0)
\rput[l]{-45}(0.145,-0.05){$q^{\star}= 0.145$}
\rput[l]{-45}(0.3,-0.05){$\delta= 0.3$}
\psdots[linecolor=black](0.145, 0)(0.145, 0.3438)(0.3, 0)(0,0.8813)(0, 0.3438)(0.3,0.1002)(1,0.1002)
\psline[linecolor=blue](0.145, 0.3438)(0.3, 0)
\rput[r](-0.04,0.8813){$H(U|Z) = 0.8813$}
\rput[r](-0.04,0.3438){$H(U|Z) - h(q^{\star}) = 0.3438$}
\psline[linestyle=dotted,linecolor=black](0.3,0)(0.3,0.1002)(1,0.1002)
\rput[l](1.04,0.1002){$H(U|Z) - h(\delta)  = 0.1002$}
\end{pspicture}
\caption{Optimal trade-off between the capacity $C$ and the optimal distortion $D_{\textsf{e}}^{\star}$ for the DSBS with parameters $p_0 = 0.5$, $\delta_0=\delta_1=0.3$, $\kappa=0$.}
\label{fig:DSBS_TradeOff}
\end{center}
\end{figure}
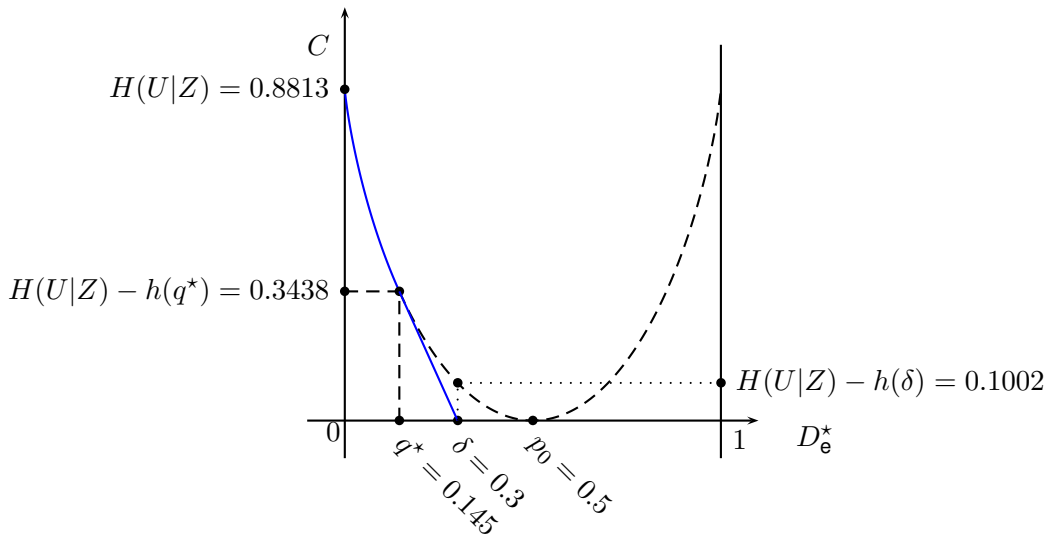

\begin{proposition}\label{prop:WynerZivDSBS}
We denote by $q^{\star}$ the unique solution of 
\begin{align}
h'(q) = \frac{H(U|Z) - h(q)}{\delta-q} .
\end{align}

1) If $C\in[0,H(U|Z) - h(q^{\star}) ]$, then the optimal splitting (Fig. \ref{fig:DSBS_02}) has three posterior beliefs\\\\
\begin{tabular}{|l|l|l|}
\hline
$q_1 = q^{\star}$ & $q_2 = \frac12 $ & $q_3 = 1-q^{\star}$\\
\hline
$\lambda_1=  \frac12 \cdot \frac{C}{H(U|Z) - h(q^{\star})}$& $\lambda_2=  1- \frac{C}{H(U|Z) - h(q^{\star})}$& $\lambda_3 =  \frac12 \cdot \frac{C}{H(U|Z) - h(q^{\star})}$\\
\hline
\end{tabular}\\\\
corresponding to the strategies\\\\
\begin{tabular}{|l|l|l|}
\hline
$\alpha_1 =  (1 - q^{\star}) \cdot \frac{C}{H(U|Z) - h(q^{\star})}$& $\alpha_2 =  1- \frac{C}{H(U|Z) - h(q^{\star})}$& $\alpha_3 =  q^{\star} \cdot \frac{C}{H(U|Z) - h(q^{\star})}$\\
\hline
$\beta_1 =  q^{\star}  \cdot \frac{C}{H(U|Z) - h(q^{\star})}$& $\beta_2 =  1- \frac{C}{H(U|Z) - h(q^{\star})}$& $\beta_3 =  (1 - q^{\star}) \cdot \frac{C}{H(U|Z) - h(q^{\star})}$\\
\hline
\end{tabular}\\\\
and the optimal distortion is
\begin{align}
D_{\textsf{e}}^{\star} =& \delta - C\cdot \frac{ \delta - q^{\star}  }{H(U|Z) - h(q^{\star})} 
\end{align}
2) If $C\in[H(U|Z) - h(q^{\star}), H(U|Z)]$, then the optimal splitting (Fig. \ref{fig:DSBS_04}) has two posterior beliefs\\\\
\begin{tabular}{|l|l|l|}
\hline
$q_1 = h^{-1}\big(H(U|Z)-C\big)$ & $q_2 = \frac12 $ & $q_3 = 1-h^{-1}\big(H(U|Z)-C\big)$\\
\hline
$\lambda_1=  \frac12 $& $\lambda_2=  0$& $\lambda_3 =  \frac12$\\
\hline
\end{tabular}\\\\
corresponding to the strategies\\\\
\begin{tabular}{|l|l|l|}
\hline
$\alpha_1 =  1 - h^{-1}\big(H(U|Z)-C\big)$& $\alpha_2 =  0$& $\alpha_3 =  h^{-1}\big(H(U|Z)-C\big)$\\
\hline
$\beta_1 =  h^{-1}\big(H(U|Z)-C\big)$& $\beta_2 =  0$& $\beta_3 =  1 - h^{-1}\big(H(U|Z)-C\big)$\\
\hline\hline
\end{tabular}\\\\
and the optimal distortion is
\begin{align}
D_{\textsf{e}}^{\star} =& h^{-1}\big(H(U|Z)-C\big),
\end{align}
where the notation $h^{-1}\big(H(U|Z)-C\big)$ stands for the unique solution of equation $h(q)=H(U|Z)-C$.\\\\
3) If $C>H(U|Z)$, then the optimal splitting rely on the two extreme posterior beliefs $(0,1)$ and $D_{\textsf{e}}^{\star} =0$.
\end{proposition}

The proof of Proposition \ref{prop:WynerZivDSBS} is given in the Appendix \ref{sec:ProofPropDSBS}. When $C\leq H(U|Z) - h(q^{\star})$, the optimal strategy consists of a time-sharing between the operating point $(D_{\textsf{e}}^{\star} ,C)=\big(q^{\star},H(U|Z)- h(q^{\star})\big)$ and the zero rate point $(\delta,0)$, as depicted in Fig. \ref{fig:DSBS_TradeOff}.

\subsection{Distinct distortions without side information, $p_0 = 0.5$, $\delta_0 = 0.5$, $\delta_1 = 0.5$, $C = 0.2$, $\kappa = \frac34$}

\begin{figure}[!ht]
\begin{center}
\psset{xunit=5cm,yunit=5cm}
\begin{pspicture}(0,-0.3)(1,1.1)
\psline{->}(0,-0.1)(0,1.1)
\psline{->}(-0.1,0)(1.1,0)
\psline{-}(1,-0.1)(1,1.1)
\rput[u](-0.03,-0.03){$0$}
\rput[l](1.05,1){$1$}
\rput[l](-0.05,1){$1$}
\rput[u](1.05,-0.05){$1$}
\rput[u](1.25,-0.05){$q$}
\psline[linestyle=dotted,linecolor=black](0, 0)(1,1)
\psline[linestyle=dotted,linecolor=black](0, 1)(1,0)
\psline(0.5,0)(0.5,1)
\rput[l]{-45}(0.5,-0.05){$ p_0 = 0.5$}
\psdots(0.5,0)
\psline[linecolor=orange, linewidth=3.5pt](0,0)(0.875,0.875)
\psline[linecolor=orange, linewidth=3.5pt](0.875,0.125)(1,0)
\rput[l]{-45}(0.875,-0.05){$\textcolor[rgb]{0,0,0}{q_2 = \gamma = 0.875}$}
\psdots(0.875 ,0)
\psplot[plotpoints=100]{0.001}{0.999}{x ln 2 ln div x mul neg 1 x neg add ln 2 ln div 1 x neg add mul neg add}
\psline[linecolor=black](0,0.8)(1, 0.8)
\psdots[linecolor=black](0.5, 0.8)(0,0.8)
\rput[r](-0.02,0.8){$H(U) - C$}
\psdots[linecolor=red,dotscale=1.5](0.875, 0.125)
\psdots[linecolor=red](0.875, 0.5435)
\psdots[linecolor=red](0.3308,  0.9157)
\psline[linecolor=red](0.3308,  0.9157)(0.875, 0.5435)
\psline[linestyle=dotted,linecolor=red](0.875,0)(0.875,0.875)
\psline[linestyle=dotted,linecolor=red](0.3308,0)(0.3308,  0.9157)
\psdots[linecolor=red,dotscale=1.5](0.3308,  0.3308)
\rput[l]{-45}(0.3308,-0.05){$q_1 = 0.3308$}
\psline[linestyle=dashed,linecolor=red](0.3308,  0.3308)(0.875, 0.125)
\psline[linestyle=dashed,linecolor=black](0.5, 0.2668)(1, 0.2668)
\psdots[linecolor=red](0.5, 0.2668)
\psdots[linecolor=black](1, 0.2668)(0.3308 ,0)
\rput[l](1.02,0.2668){$D_{\textsf{e}}^{\star} = 0.2668$}
\end{pspicture}
\caption{For the parameters $p_0 = 0.5$, $\delta_1 = \delta_2 = 0.5$, $C = 0.2$, $\kappa = \frac34$, the optimal encoder's distortion $D_{\textsf{e}}^{\star} = 0.2668$.}
\label{fig:EncoderUtility_0.5}
\end{center}
\end{figure}

We consider that the parameters $\delta_1 = \delta_2 = 0.5$ so that the side information $Z$ is independent of the source $U$. This corresponds to the problem studied in  \cite{LeTreustTomala19}, when replacing the minimization by the maximization. We have $H_b(\delta_0) = H_b(\delta_1)=H_b\Big((1 - q) \cdot \delta_0+ q \cdot (1 - \delta_1)\Big)=1$ and $\nu_1 = \nu_2 = \gamma = \frac78$, the average entropy and average distortion write
\begin{align}
h(q) =& H_b(q),\\
\Psi_{\textsf{e}}(q) = &\psi_{\textsf{e}}(q) = p \cdot  \UN\big(p\leq \gamma\big) + (1-p) \cdot  \UN\big(p >\gamma\big).
\end{align}
Applying \cite[Corollary 3.5, pp. 15]{LeTreustTomala19}, the optimal splitting has two posteriors $|\mc{W}|=2$ and must satisfy the information constraint
\begin{align}
&\frac{p_0 - q_2}{q_1-q_2}\cdot H_b(q_1) +  \frac{q_1 - p_0}{q_1-q_2} \cdot H_b(q_2) \geq H(U) - C\label{eq:IC-noSide}.
\end{align}
The distortion function has two piece-wise linear components, hence the optimal splitting involves $q_1\in[0,p_0]$ and $q_2\in[\gamma,1]$. For each $q_2\in[\gamma,1]$, we denote by $q_1(q_2)$ the function that returns the posterior which satisfies \eqref{eq:IC-noSide} with equality. From \cite[Fig. 5, pp. 19]{LeTreustTomala19}, the function $q_1(q_2)$ is strictly increasing, hence its derivative $q_1'(q_2)$ is strictly positive. The encoder's distortion function reformulates in terms of $q_2$ as
\begin{align}
\Phi_{\textsf{e}}(q_2)=&\frac{p_0 - q_2}{q_1(q_2)-q_2}\cdot q_1(q_2) +  \frac{q_1 - p_0}{q_1(q_2)-q_2} \cdot(1-q_2).
\end{align}
Its derivative writes
\begin{align}
\Phi_{\textsf{e}}'(q_2)=&\frac{1}{(q_1(q_2)-q_2)^2}\cdot \bigg( q_1'(q_2)\cdot  \Big(q_2 \cdot (2 \cdot q_2 - 1) +p_0\Big) - (p_0 - q_1(q_2)) \cdot (1 - 2 \cdot q_1(q_2)) \bigg).
\end{align}
Since $q_2\geq\gamma>\frac12$, the sign of the derivative is negative if and only if
\begin{align}
0< q_1'(q_2)\leq \frac{(p_0 - q_1(q_2)) \cdot (1 - 2 \cdot q_1(q_2))}{q_2 \cdot (2 \cdot q_2 - 1) +p_0}.
\end{align}
By numerical optimization, the above inequality is satisfied for $p_0 = 0.5$, $\delta_1 = \delta_2 = 0.5$, $C = 0.2$, $\kappa = \frac34$, hence the optimal distortion is achieved by using $q_2 = \gamma$, as depicted on Fig. \ref{fig:EncoderUtility_0.5}.


\begin{figure}[!ht]
\begin{center}
\psset{xunit=5cm,yunit=5cm}
\begin{pspicture}(0,-0.4)(1,1.1)
\psline{->}(0,-0.1)(0,1.1)
\psline{->}(-0.1,0)(1.1,0)
\psline{-}(1,-0.1)(1,1.1)
\rput[u](-0.03,-0.03){$0$}
\rput[l](-0.05,1){$1$}
\rput[u](1.05,-0.05){$1$}
\rput[u](1.17,-0.05){$q$}
\psdots(0.5,0)
\psline(0.5,0)(0.5,1)
\rput[l]{-45}(0.54,-0.05){$ p_0 = 0.5$}
\psdots(0.4178,0)(0.9301,0)
\psline(0.4178,0)(0.4178,1)
\psline(0.93,0)(0.9301,1)
\rput[l]{-45}(0.9,-0.05){$  q_3^{\star} = \nu_1 = 0.9301$}
\rput[l]{-45}(0.4178,-0.05){$  q_2^{\star} = \nu_2 = 0.4178$}
\psline[linecolor=orange, linewidth=3.5pt](0,0)(0.4178, 0.4178)
\psline[linecolor=orange, linewidth=3.5pt](1,0)(0.9301, 0.0699)
\psline[linecolor=orange, linewidth=3.5pt](0.4118,0.2353)(0.9301, 0.4685)
\PlotConditionalEntropy{0.5}{0.05}{0.5}
\psline[linecolor=black](0,0.5947)(1, 0.5947)
\psdots[linecolor=black](0.5, 0.5947)(0,0.5947)(0.0715, 0)
\rput[l]{-45}(0.0715,-0.05){$ q_1^{\star} = 0.0715$}
\rput[r](-0.02,0.5947){$H(U|Z) - C = 0.5947$}
\psdots[linecolor=red,dotscale=1.5] (0.4118,0.2353)(0.0715, 0.0715)(0.9301, 0.0699)
\psdots[linecolor=red](0.4118,  0.7705)(0.0715, 0.2991)(0.9301, 0.3185)
\psline[linecolor=red](0.4118,  0.7705)(0.0715, 0.2991)(0.9301, 0.3185)(0.4118,  0.7705)
\psline[linecolor=red,linestyle=dotted](0.0715, 0.2991)(0.0715, 0)
\psline[linestyle=dashed,linecolor=red](0.9301, 0.0699)(0.4118,  0.2353)(0.0715, 0.0715)(0.9301, 0.0699)

\psline[linestyle=dashed,linecolor=black](0.5, 0.1721)(0, 0.1721)
\psdots[linecolor=red](0.5, 0.1721)
\psdots[linecolor=black](0, 0.1721)
\rput[r](-0.02,0.1721){$D_{\textsf{e}}^{\star} = 0.1721$}

\end{pspicture}
\caption{For the parameters $p_0 = 0.5$, $\delta_0 = 0.05$, $\delta_1 = 0.5$, $C = 0.2$, $\kappa = \frac34$, the optimal encoder's distortion is $D_{\textsf{e}}^{\star}=0.1721$.}
\label{fig:EncoderUtility}
\end{center}
\end{figure}
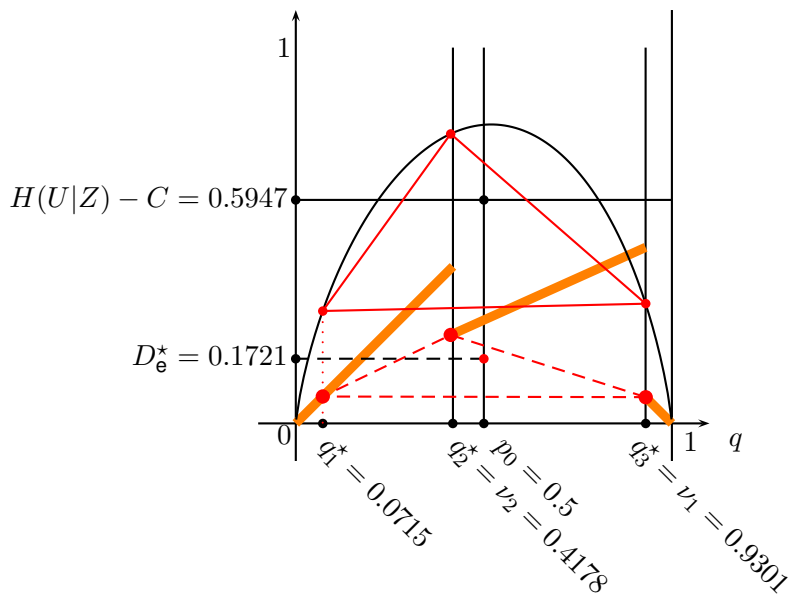

\subsection{Distinct distortions with side information, $p_0 = 0.5$, $\delta_0 = 0.05$, $\delta_1 = 0.5$, $C = 0.2$, $\kappa = \frac34$}

We consider an example with distinct distortion functions, i.e. with $\kappa = \frac34$, with decoder's side information. By numerical simulation, we determine the optimal triple of posteriors $(q_1 ,q_2 ,q_3 )$ represented by the red dots in Fig. \ref{fig:EncoderUtility}, that corresponds to the minimal distortion $D_{\textsf{e}}^{\star}=0.1721$.\\\\
\begin{tabular}{|l|l|l|}
\hline
$q_1 = 0.0715$ & $q_2 = 0.4118$ & $q_3 = 0.9301$\\
\hline
$\lambda_1 =  0.1288$& $\lambda_2 = 0.6165$&$ \lambda_3 =  0.2548$\\
\hline
\end{tabular}\\

The parameters of the optimal strategy in Fig. \ref{fig:SignalingZ}, are given by\\\\
\begin{tabular}{|l|l|l|}
\hline
$\alpha_1 = 0.2392$ & $\alpha_2 = 0.7252$ & $\alpha_2 = 0.0356$\\
\hline
$\beta_1 =  0.0184$& $\beta_2 = 0.5077$&$ \beta_3 =  0.4739$\\
\hline
\end{tabular}



%
%
%
%


\appendices



\section{Proof of Theorem \ref{theo:Concavification}}\label{sec:ProofConcavification}

We consider a joint probability distribution $\QQ_{UW}\in\Delta(\mc{U}\times \mc{W})$, we identify the parameters $\lambda_w =\QQ(w)$ and $p_w = \QQ_U(\cdot|w)\in \Delta(\mc{U})$. The average distortion writes
\begin{align}
\sum_{w} \lambda_w \cdot \Psi_{\textsf{e}}(p_w)  
=& \sum_{w} \QQ(w) \cdot \Psi_{\textsf{e}}\big(\QQ_U(\cdot|w)\big) \label{eq:ConvexificationDistortion1}\\
=& \sum_{w} \QQ(w) \cdot  \sum_{u,z}  \QQ(u|w)  \cdot \PP(z|u)\cdot  \psi_{\textsf{e}}\Big(z,\QQ_U(\cdot|w,z) \Big)\label{eq:ConvexificationDistortion2} \\
=& \sum_{w,z} \QQ(w) \cdot  \QQ(z|w) \cdot  \E_{\QQ_U(\cdot|w,z) } \bigg[d_{\textsf{e}}\Big(U,z,v^{\star}\big(z,\QQ_U(\cdot|w,z)\big)\Big)\bigg] \label{eq:ConvexificationDistortion3} \\
=& \E_{\QQ_{UZW} } \bigg[d_{\textsf{e}}\Big(U,Z,V^{\star}\big(Z,\QQ_U(\cdot|W,Z)\big)\Big)\bigg] \label{eq:ConvexificationDistortion4} \\
=& \max_{\QQ_{V|ZW}  \in  \atop \Q_2(\QQ_{UZW})} \E_{\QQ_{UZW} \atop \QQ_{V|ZW}} \bigg[d_{\textsf{e}}(U,Z,V)\bigg] .\label{eq:ConvexificationDistortion5}
\end{align}
Equations \eqref{eq:ConvexificationDistortion2}, \eqref{eq:ConvexificationDistortion3} and \eqref{eq:ConvexificationDistortion5} come from Definitions \ref{def:AverageUtility}, \ref{def:RobustUtility} and \ref{def:BestReply}.\\
Equations  \eqref{eq:ConvexificationDistortion1} and  \eqref{eq:ConvexificationDistortion4} are reformulations. 

The average entropy writes
\begin{align}
\sum_{w} \lambda_w \cdot h(p_w)
=& \sum_{w} \QQ(w) \cdot h\big(\QQ_U(\cdot|w)\big)\label{eq:ConvexificationEntropy1}\\
=& \sum_{w} \QQ(w) \cdot \bigg( H\big(\QQ_U(\cdot|w)\big) + \sum_{u} \QQ(u|w) \cdot H\Big(\PP_Z(\cdot|u)\Big) - H\Big( \sum_{u}  \QQ(u|w) \cdot\PP_Z(\cdot|u)\Big) \bigg)\label{eq:ConvexificationEntropy2}\\
=& \sum_{w} \QQ(w) \cdot \bigg( H\big(\QQ_U(\cdot|w)\big) + \sum_{u} \QQ(u|w) \cdot H\Big(\QQ_Z(\cdot|u,w)\Big) - H\Big( \QQ_Z(\cdot|w)\Big) \bigg)\label{eq:ConvexificationEntropy3}\\
=& H(U|W) + H(Z|U,W) - H(Z|W) = H(U|W,Z).\label{eq:ConvexificationEntropy4}
\end{align}
Equation \eqref{eq:ConvexificationEntropy2} come from Definition \ref{def:AverageUtility}.\\
Equation \eqref{eq:ConvexificationEntropy3} come from Markov chain property $Z  -\!\!\!\!\minuso\!\!\!\!-U    -\!\!\!\!\minuso\!\!\!\!-  W$ that implies $\PP_Z(\cdot|u)=\QQ_Z(\cdot|u,w)$ and $\QQ_Z(\cdot|w) =  \sum_{u}  \QQ(u|w) \cdot\PP_Z(\cdot|u)$.\\
Equations  \eqref{eq:ConvexificationEntropy1} and  \eqref{eq:ConvexificationEntropy4} are reformulations. 

Hence, equation \eqref{eq:SplittingFormulation} reformulates
\begin{align}
&\inf_{\lambda_w\in [0,1],\atop p_w\in \Delta(\mc{U})}\bigg\{ \sum_{w} \lambda_w \cdot \Psi_{\textsf{e}}(p_w)   \quad \text{ s.t. } \quad \sum_{w} \lambda_w \cdot p_w  = \PP_U \in \Delta(\mc{U}), \nonumber \\
&\qquad\qquad\qquad\qquad\qquad \text{ and }\quad \sum_{w} \lambda_w \cdot h(p_w)  \geq H(U|Z)  - \max_{\PP_X}I(X;Y) \bigg\}\label{eq:SplittingFormulation1}\\
=&\inf_{ \QQ_W , \QQ_{U|W}}\bigg\{ \max_{\QQ_{V|ZW}  \in  \atop \Q_2(\QQ_{UZW})} \E_{\QQ_{UZW} \atop \QQ_{V|ZW}} \Big[d_{\textsf{e}}(U,Z,V)\Big]   \quad \text{ s.t. } \quad \sum_{w} \QQ(w) \cdot \QQ_U(\cdot|w)  = \PP_U \in \Delta(\mc{U}),\nonumber\\
&\qquad\qquad\qquad\qquad\qquad \text{ and }\quad H(U|W,Z)  \geq H(U|Z)  - \max_{\PP_X}I(X;Y) \bigg\}\label{eq:SplittingFormulation2}\\
=&\inf_{ \QQ_{UZW} \in \Q_0} \max_{\QQ_{V|ZW}  \in  \atop \Q_2(\QQ_{UZW})} \E_{\QQ_{UZW} \atop \times \QQ_{V|ZW}} \bigg[d_{\textsf{e}}(U,Z,V)\bigg] = D_{\textsf{e}}^{\star}.\label{eq:SplittingFormulation7}
\end{align}
This concludes the proof of Theorem \ref{theo:Concavification}.


\section{Achievability proof of Theorem \ref{theo:MaxMinStackelberg}}\label{sec:AchievabilityProof}

We refine the analysis of the Wyner-Ziv's source encoding scheme \cite{wyner-it-1976}, in order to control the posterior beliefs of a large fraction of stages, as stated in Proposition \ref{prop:WynerZivCoding} and in \eqref{eq:SetTalpha}-\eqref{eq:SetBalpha}. Corollary \ref{coro:AchievabilityProof} shows that the best-reply strategy of the decoder performs similarly as the Wyner-Ziv's decoding scheme.

%
%
%
%
%
%
%
%
%
%
%

%
%


\subsection{Zero capacity}\label{sec:ZeroCapacity}

We first investigate the special case of zero capacity.
\begin{lemma}\label{lemma:ZeroCapacity}
If the channel has zero capacity $\max_{\PP_X} I( X; Y ) =0$, then  we have:
\begin{align}
\forall n \in \N^{\star},\; \forall \sigma,\qquad  \max_{\tau \in \textsf{BR}_{\textsf{d}}(\sigma)} d_{\textsf{e}}^{\,n}(\sigma, \tau)  =  D_{\textsf{e}}^{\star}.
\end{align}
\end{lemma}
\begin{proof}[Lemma \ref{lemma:ZeroCapacity}]
When capacity is zero $\max_{\PP_X} I( X; Y ) =0$, then the probability distribution $ \PP_{UZ}  \QQ_{W|U} \in \Q_0$ must satisfy $I( U ;W |Z ) =0$, hence the Markov chain property $U -\!\!\!\!\minuso\!\!\!\!- Z -\!\!\!\!\minuso\!\!\!\!- W$, i.e. $\QQ_{U|ZW}=\PP_{U|Z}$. 
\begin{align}
D_{\textsf{e}}^{\star}=&\inf_{ \QQ_{UZW} \in \Q_0} \max_{\QQ_{V|WZ}  \in  \atop \Q_2(\QQ_{UZW})} \E_{\QQ_{UZW} \atop  \QQ_{V|WZ}} \bigg[d_{\textsf{e}}(U,Z,V)\bigg] \label{eq:LemmaZero1}\\
=& \inf_{ \QQ_{UZW} \in \Q_0}\E_{\QQ_{UZW} } \bigg[d_{\textsf{e}}\Big(U,Z,V^{\star}\big(Z,\QQ_U(\cdot|W,Z)\big)\Big)\bigg]\label{eq:LemmaZero2}\\
=& \inf_{ \QQ_{UZW} \in \Q_0}\E_{\QQ_{UZW} } \bigg[d_{\textsf{e}}\Big(U,Z,V^{\star}\big(Z,\PP_U(\cdot|Z)\big)\Big)\bigg]\label{eq:LemmaZero3}\\
=& \E_{\PP_{UZ}} \bigg[d_{\textsf{e}}\Big(U,Z,V^{\star}\big(Z,\PP_U(\cdot|Z)\big)\Big)\bigg]\label{eq:LemmaZero4}.
\end{align}
Equation \eqref{eq:LemmaZero2} is a reformulation by using the best-reply $v^{\star}\big(z,p\big)$ of Definition \ref{def:BestReply} for symbol $z\in \mc{Z}$ and the belief $\QQ_U(\cdot|w,z)$.\\
Equation \eqref{eq:LemmaZero3} comes from Markov chain property $U -\!\!\!\!\minuso\!\!\!\!- Z -\!\!\!\!\minuso\!\!\!\!- W$ that allows to replace the belief $\QQ_U(\cdot|w,z)$ by $\PP_U(\cdot|z)$.\\
Equation \eqref{eq:LemmaZero4} comes from removing the random variable $W$ since it has no impact on the distortion function $d_{\textsf{e}}\Big(u,z,v^{\star}\big(z,\PP_U(\cdot|z)\big)\Big)$.\\ 

For any $n$ and for any encoding strategy $\sigma$, the encoder's long-run distortion is given by
\begin{align}
\max_{\tau \in \textsf{BR}_{\textsf{d}}(\sigma)}d_{\textsf{e}}^{\,n}(\sigma, \tau) 
&=\max_{\tau \in \textsf{BR}_{\textsf{d}}(\sigma)}\sum_{u^n,z^n,x^n,\atop y^n,v^n}\prod_{t=1}^n\PP\big(u_t,z_t \big)  \sigma\big(x^n\big| u^n \big) \prod_{t=1}^n \mc{T}\big(y_t\big)  \tau\big(v^n \big| y^n ,z^n\big) \cdot  \Bigg[  \frac{1}{n} \sum_{t=1}^n d_{\textsf{e}}(u_t,z_t,v_t) \Bigg]\label{eq:LemmaZero5}\\
&=\max_{\tau \in \textsf{BR}_{\textsf{d}}(\sigma)}\sum_{u^n,z^n,v^n}\prod_{t=1}^n\PP\big(u_t,z_t \big)   \tau\big(v^n \big|z^n\big) \cdot  \Bigg[  \frac{1}{n} \sum_{t=1}^n d_{\textsf{e}}(u_t,z_t,v_t) \Bigg]\label{eq:LemmaZero6}\\
&= \frac{1}{n} \sum_{t=1}^n \Bigg[ \sum_{u_t,z_t,v_t} \PP\big(u_t,z_t \big) \cdot  \UN\Big(v_t = v^{\star}\big(z_t,\QQ_U(\cdot|z_t)\big)\Big) \cdot   d_{\textsf{e}}(u_t,z_t,v_t)\Bigg] \label{eq:LemmaZero7}\\
&=  \E_{\PP(u,z) } \bigg[d_{\textsf{e}}\Big(U,Z,V^{\star}\big(Z,\PP_U(\cdot|Z)\big)\Big)\bigg].\label{eq:LemmaZero9}
\end{align}
Equation \eqref{eq:LemmaZero5} comes from the zero capacity which imposes that the channel outputs $Y^n$ are independent of the channel inputs $X^n$.\\
Equation \eqref{eq:LemmaZero6} comes from removing the random variables $(X^n,Y^n)$ and noting that the decoder's best-reply $ \tau\big(v^n \big|z^n\big)$ does not depend on $y^n$ anymore, since $y^n$ is independent of $(u^n,z^n)$.\\
Equation \eqref{eq:LemmaZero7} is a reformulation based on the best-reply $v^{\star}\big(z,\PP_U(\cdot|z)\big)$ of Definition \ref{def:BestReply}, for the symbol $z\in \mc{Z}$ and the belief $\PP_U(\cdot|z)$.\\
Equation \eqref{eq:LemmaZero9} comes from the i.i.d. property of $(U,Z)$ and  concludes the proof of Lemma \ref{lemma:ZeroCapacity}.
\end{proof}


\subsection{Strictly positive capacity}\label{sec:PositiveCapacity}

We now assume that the channel capacity is strictly positive $\max_{\PP_X} I( X; Y ) >0$. We define a specific convex closure in which the information constraint is satisfied with \emph{strict} inequality and the sets of decoder's best-reply symbols are always \emph{singletons}.
\begin{align}
\widehat{D_{\textsf{e}}}  =&  \inf\bigg\{ \sum_{w} \lambda_w \cdot \Psi_{\textsf{e}}(p_w)   \quad \text{ s.t. } \quad \sum_{w} \lambda_w \cdot p_w  = \PP_U \in \Delta(\mc{U}),\nonumber\\
&\qquad\qquad\qquad\qquad\qquad \text{ and }\quad \sum_{w} \lambda_w \cdot h(p_w)  > H(U|Z)  - \max_{\PP_X}I(X;Y),\nonumber\\
&\qquad\qquad\qquad\qquad\qquad \text{ and }\quad \forall (z,w)\in \mc{Z}\times\mc{W}, \;\; \mc{V}^{\star}\big(z,\QQ_U(\cdot|z,w)\big) \;\; \text{is a singleton} \;\; \bigg\}.\label{eq:SplittingFormulationHat}
\end{align}

\begin{lemma}\label{lemma:RestrictedSplitting}
If $\max_{\PP_X} I( X; Y ) >0$, then $\widehat{D_{\textsf{e}}} = D_{\textsf{e}}^{\star}$.
\end{lemma}
For the proof of Lemma \ref{lemma:RestrictedSplitting}, we refers to the similar proof of \cite[Lemma A.5, pp.  32]{LeTreustTomala19}. We denote by  $\PP^{\star}_X$ the probability distribution that maximizes the mutual information $I(X;Y)$, we denote by $Q^n_{UZW}$ the empirical distribution of the sequence $(u^n,z^n,w^n)$ and we denote by $A_{\delta}$ the set of typical sequences with tolerance $\delta>0$, defined by
\begin{align}
A_{\delta} = \bigg\{(u^n,z^n,w^n,x^n,y^n), &\quad\text{ s.t. } \quad|| Q^n_{UZW} - \PP_{UZ} \QQ_{W|U}||_1\leq \delta,\nonumber\\
 &\quad\text{ and } \quad|| Q^n_{XY} - \PP^{\star}_X  \mc{T}_{Y|X}||_1\leq \delta  \bigg\}.
\end{align}
We denote by $\PP_{\sigma,U_t}(\cdot|y^n,z^n)\in \Delta(\mc{U})$ the posterior belief induced by the strategy $\sigma$ on $U_t$ at stage $t$, given $(y^n,z^n)$. We define $T_{\alpha}(w^n,y^n,z^n)$ and $B_{\alpha,\gamma,\delta}$ depending on parameters $\alpha>0$ and $\gamma>0$:
\begin{align}
T_{\alpha}(w^n,y^n,z^n) =& \bigg\{t \in \{1,\ldots,n\} , \;\;\text{ s.t. }\;\; D\Big(\PP_{\sigma,U_t}(\cdot|y^n,z^n)\Big|\Big|\QQ_{U_t}(\cdot|w_t,z_t)\Big)\leq  \frac{\alpha^2}{2\ln 2} \bigg\},\label{eq:SetTalpha}\\
B_{\alpha,\gamma,\delta}=& \bigg\{ (w^n,y^n,z^n) , \;\;\text{ s.t. } \;\; \frac{|T_{\alpha}(w^n,y^n,z^n)|}{n} \geq 1 - \gamma\;\; \text{ and }\;\;(w^n,y^n,z^n) \in A_{\delta}\bigg\}.\label{eq:SetBalpha}
\end{align}
The notation $B_{\alpha,\gamma,\delta}^c$ stands for the complementary set of $B_{\alpha,\gamma,\delta}\subset \mc{W}^n \times \mc{Y}^n \times \mc{Z}^n$. The sequences $(w^n,y^n,z^n)$ belong to the set $B_{\alpha,\gamma,\delta}$ if, 1) they are typical, 2) the corresponding posterior belief $\PP_{\sigma,U_t}(\cdot|y^n,z^n)$ is close in K-L divergence to the target belief $\QQ_{U_t}(\cdot|w_t,z_t)$, for a large fraction of stages $t \in \{1,\ldots,n\}$.

The cornerstone of this achievability proof is Proposition \ref{prop:WynerZivCoding}, which refines the analysis of Wyner-Ziv's source coding by controlling the posterior beliefs of a large fraction of stages.

\begin{proposition}[Wyner-Ziv's Posterior Beliefs]\label{prop:WynerZivCoding}
If the probability distribution $\PP_{UZ} \QQ_{W|U}$ satisfies:
\begin{align}
\begin{cases}
&\max_{\PP_X}I(X;Y) - I(U;W|Z) >0,\\
&\mc{V}^{\star}\big(z,\QQ_U(\cdot|z,w)\big) \text{ is a singleton }\forall (z,w)\in \mc{Z}\times\mc{W},
\end{cases}
\end{align}
then
\begin{align}
&\forall \varepsilon>0, \;\forall \alpha>0, \;\forall \gamma>0,\;\exists \bar{\delta}>0,\;\forall \delta< \bar{\delta}, \;\exists \bar{n}\in \N^{\star},\;\forall n\geq  \bar{n},  \exists \sigma, \text{ s.t. } \PP_{\sigma}(B_{\alpha,\gamma,\delta}^c) \leq \varepsilon.\label{eq:propWynerZivCoding}
\end{align}
\end{proposition}
The proof of proposition \ref{prop:WynerZivCoding} is stated in App. \ref{sec:ProofPropositionCode}.
\begin{proposition}\label{prop:UtilityBound}
For any encoding strategy $\sigma$, we have:
\begin{align}
\bigg| \max_{\tau \in \textsf{BR}_{\textsf{d}}(\sigma)}d_{\textsf{e}}^{\,n}(\sigma, \tau) - \widehat{D_{\textsf{e}}}\bigg| \leq (\alpha+ 2 \gamma + \delta)\cdot \bar{d_{\textsf{e}}} + (1 - \PP_{\sigma}(B_{\alpha,\gamma,\delta})) \cdot \bar{d_{\textsf{e}}},
\end{align}
where $ \bar{d_{\textsf{e}}} = \max_{u,z,v} \big|d_{\textsf{e}}(u,z,v)\big| $ is the largest absolute value of encoder's distortion.
\end{proposition}
For the proof of Proposition \ref{prop:UtilityBound}, we refers directly to the similar proof of \cite[Lemma A.8, pp. 33]{LeTreustTomala19}. 

\begin{corollary}\label{coro:AchievabilityProof}
For any $\varepsilon>0$, there exists $\bar{n}\in \N^{\star}$ such that for all $n\geq \bar{n}$ there exists an encoding strategy $\sigma$ such that:
\begin{align}
\bigg| \max_{\tau \in \textsf{BR}_{\textsf{d}}(\sigma)}d_{\textsf{e}}^{\,n}(\sigma, \tau) - \widehat{D_{\textsf{e}}}\bigg| \leq \varepsilon.
\end{align}
\end{corollary}
The proof of Corollary \ref{coro:AchievabilityProof} comes from combining Proposition \ref{prop:WynerZivCoding} with Proposition \ref{prop:UtilityBound} and choosing  parameters $\alpha$, $\gamma$, $\delta$ small and  $n\in \N^{\star}$ large. The decoder's best-reply performs similarly as the Wyner-Ziv's decoding scheme. It concludes the achievability proof of Theorem \ref{theo:MaxMinStackelberg}.


\subsection{Proof of Proposition \ref{prop:WynerZivCoding}}\label{sec:ProofPropositionCode}

We assume that the probability distribution $\PP_{UZ} \QQ_{W|U}$ satisfies the two following conditions:
\begin{align}
\begin{cases}
&\max_{\PP_X}I(X;Y) - I(U;W|Z) >0,\\
&\mc{V}^{\star}\big(z,\QQ_U(\cdot|z,w)\big) \text{ is a singleton }\forall (z,w)\in \mc{Z}\times\mc{W}.
\end{cases}
\end{align}
The strict information constraint ensures that  there exists a small parameter $\eta>0$ and rates  $\textsf{R}\geq 0 $, $\textsf{R}_{\textsf{L}}\geq 0 $, such that
\begin{eqnarray}
\textsf{R}  + \textsf{R}_{\textsf{L}}& =&       I( U;W )  + \eta  \label{eq:AchievabilityB1} , \\
\textsf{R}_{\textsf{L}}  &\leq &       I( Z;W )  - \eta  \label{eq:AchievabilityB2} , \\
\textsf{R} \; & \leq&   \max_{\PP_X} I( X; Y )  -  \eta  \label{eq:AchievabilityB3}  .
\end{eqnarray}
We now recall the random coding constructions of Wyner-Ziv and Shannon for the source and the channel, in \cite{wyner-it-1976} and  \cite{shannon-bell-1948}. Then, we investigate the decoder's posterior beliefs regarding the sequence of symbols of source. 
We denote by $\Sigma$ the random coding scheme, described as follows.
\begin{itemize}
\item[$\bullet$] \textit{Random codebook.} We introduces the indices $m\in\mc{M}$ with $| \mc{M}|= 2^{n   \sf{R}} $ and $l\in\mc{M}_{\textsf{L}}$ with $| \mc{M}_{\textsf{L}}  |= 2^{n   \sf{R}_{\textsf{L}} } $. We draw $| \mc{M} \times \mc{M}_{\textsf{L}}   |= 2^{n  ( \sf{R}  + \textsf{R}_{\textsf{L}})  } $ sequences $W^n(m,l)$ with the i.i.d. probability distribution $\QQ^{\otimes n}_W $, and $| \mc{M}  |= 2^{n   \sf{R}   } $ sequences $X^n(m)$,  with the i.i.d. probability distribution $\PP^{\star \otimes n}_X $ that maximizes the channel capacity in \eqref{eq:AchievabilityB3}. 
\item[$\bullet$] \textit{Encoding function.} The encoder observes the sequence of symbols of source $U^n \in  \mc{U}^n$ and finds a pair of indices $(m,l)\in \mc{M} \times  \mc{M}_{\textsf{L}}$ such that the sequences  $\big(U^n,W^n(m,l)\big) \in A_{\delta}$ are jointly typical. It sends the sequence $X^n(m)$ corresponding to the index $m\in \mc{M}$.
\item[$\bullet$] \textit{Decoding function.} The decoder observes the sequence of channel output $Y^n\in\mc{Y}^n$. It returns an index $\hat{m}\in \mc{M}$ such that the sequences  $\big(Y^n,X^n(\hat{m})\big) \in A_{\delta}$ are jointly typical.  Then it observes the sequence of side information $Z^n\in\mc{Z}^n$ and returns an index $\hat{l}\in \mc{M}_{\textsf{L}}$ such that the sequences  $\big(Z^n,W^n(\hat{m},\hat{l})\big) \in A_{\delta}$ are jointly typical. \item[$\bullet$] \textit{Error Event.} We introduce the event of error $E_{\delta} \in \{0,1\}$ defined as follows:
\end{itemize}
\begin{align}
E_{\delta} = \Bigg\{
\begin{array}{lll}
0&\text{ if }&  (M,L)=( \hat{M},\hat{L})  \;\; \text{ and }\;  \big(U^n ,  Z^n, W^n, X^n,Y^n \big)    \in A_{\delta} ,\\
1 &\text{ otherwise.}& 
\end{array}
\Bigg.
\end{align}

\textit{Expected error probability of the random coding scheme $\Sigma$.}
For all $\varepsilon_2>0$, for all $\eta >0$, there exists a $\bar{\delta}>0$, for all $\delta \leq \bar{\delta}$  there exists $\bar{n}$ such that for all $n\geq\bar{n}$, the expected probability of the following error events are bounded by $\varepsilon_2$:
\begin{align}
&\E_{\Sigma}\bigg[ \PP\bigg( \forall  (m,l) ,\quad 
\big(U^n, W^n(m,l) \big) \notin A_{\delta} \bigg)\bigg]  \leq \varepsilon_2, \label{eq:AchievProbaB1} \\
&\E_{\Sigma}\bigg[ \PP\bigg(  \exists l'\neq  l  ,\text{ s.t. } 
\big(Z^n , W^n(m,l') \big) \in A_{\delta}\bigg)\bigg]   \leq \varepsilon_2, \label{eq:AchievProbaB2}\\
&\E_{\Sigma}\bigg[ \PP\bigg(  \exists m'\neq  m  ,\text{ s.t. } 
\big(Y^n , X^n(m') \big) \in A_{\delta}\bigg)\bigg]   \leq \varepsilon_2, \label{eq:AchievProbaB3}
\end{align}
Equation \eqref{eq:AchievProbaB1} comes from \eqref{eq:AchievabilityB1} and the covering lemma \cite[pp. 208]{ElGammalKim(book)11}.\\
Equation \eqref{eq:AchievProbaB2} comes from \eqref{eq:AchievabilityB2} and the packing lemma \cite[pp. 46]{ElGammalKim(book)11}.\\
Equation \eqref{eq:AchievProbaB3} comes from \eqref{eq:AchievabilityB3} and the packing lemma \cite[pp. 46]{ElGammalKim(book)11}.

We conclude that
\begin{align}
&\forall \varepsilon_2>0,\;  \forall \eta>0, \;  \exists \bar{\delta}>0,\;\forall \delta\leq \bar{\delta},  \; \exists \bar{n}>0,\;\forall n\geq \bar{n},\quad \exists \sigma,\qquad  \PP_{\sigma}\big(E_{\delta}=1 \big) \leq \varepsilon_2. \label{eq:BoundError0}
\end{align}

\textit{Control of the posterior beliefs.} We assume that the event $E_{\delta}=0$ is realized. We denote by $\PP_{\sigma,U_t}(\cdot|y^n,z^n,E_{\delta}=0)$ the conditional probability distribution of $U_t$ given $(y^n,z^n,E=0)$, induced by the encoding strategy $\sigma$ obtained by the concatenation of Wyner-Ziv's encoding scheme and Shannon's channel encoding scheme.
\begin{align}
&  \E_{\sigma} \Bigg[ \frac{1}{n}  \sum_{t=1}^n D\bigg(  \PP_{\sigma,U_t}(\cdot|Y^n,Z^n,E_{\delta}=0) \bigg| \bigg|   \QQ_{U_t}(\cdot|W_t,Z_t) \bigg)\Bigg] \nonumber \\
=& \sum_{(w^n,z^n,y^n)\in A_{\delta} }\PP_{\sigma}(w^n,z^n,y^n|E_{\delta}=0) \times \frac{1}{n}  \sum_{t=1}^n D\bigg(  \PP_{\sigma,U_t}(\cdot|y^n,z^n,E_{\delta}=0) \bigg| \bigg|   \QQ_{U_t}(\cdot|w_t,z_t) \bigg) \label{eq:Beliefs1} \\
=&\frac{1}{n}   \sum_{(u^n,z^n,w^n,y^n)\in A_{\delta} }\PP_{\sigma}(u^n,z^n,w^n,y^n|E_{\delta}=0)    \times \log_2 \frac{1}{\prod_{t=1}^n \QQ(u_t|w_t,z_t)}  -   \frac{1}{n}  \sum_{t=1}^n H(U_t|Y^n,Z^n,E_{\delta}=0) \nonumber \\&& \label{eq:Beliefs2} \\
\leq&H(U|W,Z)   - \frac{1}{n}  H(U^n|W^n,Y^n,Z^n,E_{\delta}=0) + \delta  \label{eq:Beliefs3} \\
\leq&H(U|W,Z)  - \frac{1}{n}  H(U^n|W^n,Z^n,E_{\delta}=0) + \delta \label{eq:Beliefs4} \\
=&H(U|W,Z) -  \frac{1}{n}  H(U^n|E_{\delta}=0) +  \frac{1}{n}  I(U^n;W^n|E_{\delta}=0) \nonumber\\
+& \frac{1}{n}  H(Z^n|W^n,E_{\delta}=0) - \frac{1}{n}  H(Z^n|U^n,W^n,E_{\delta}=0) + \delta. \label{eq:Beliefs5} 
\end{align}
Equation \eqref{eq:Beliefs1}-\eqref{eq:Beliefs2} come from the hypothesis $E_{\delta}=0$ of typical sequences $(u^n,z^n,w^n,y^n)\in A_{\delta} $ and the definition of the conditional K-L divergence \cite[pp. 24]{cover-book-2006}.\\
Equation \eqref{eq:Beliefs3} comes from property of typical sequences \cite[pp. 26]{ElGammalKim(book)11} and the conditioning that reduces entropy.\\
Equation \eqref{eq:Beliefs4} comes from the Markov chain $Z^n -\!\!\!\!\minuso\!\!\!\!- U^n -\!\!\!\!\minuso\!\!\!\!- W^n -\!\!\!\!\minuso\!\!\!\!- Y^n$ induced by the channel and the strategy $\sigma$, that implies $H(U^n|W^n,Z^n,E_{\delta}=0) =H(U^n|W^n,Y^n,Z^n,E_{\delta}=0)$.\\
Equation \eqref{eq:Beliefs5} is a reformulation of \eqref{eq:Beliefs4}.

We denote by $A_{\delta}(z^n|w^n)$ the set of sequences $z^n\in\mc{Z}^n$ that are jointly typical with $w^n$, i.e. $||Q_{WZ}^n - \QQ_{WZ}||_1< \delta$. 
\begin{align}
\frac{1}{n}  H(U^n|E_{\delta}=0)\geq& H(U)  - \frac{1}{n}   - \log_2 |\mc{U}| \cdot \PP_{\sigma}\big(E_{\delta}=1 \big),  \label{eq:Control1}\\
\frac{1}{n}  I(U^n;W^n|E_{\delta}=0) \leq& \textsf{R}  + \textsf{R}_{\textsf{L}} =       I( U;W )  + \eta, \label{eq:Control2}\\
\frac{1}{n}  H(Z^n|W^n,E_{\delta}=0) \leq& \frac{1}{n} \log_2 |A_{\delta}(z^n|w^n)| \leq  H(Z|W) + \delta, \label{eq:Control3}\\
\frac{1}{n}   H(Z^n|U^n,W^n,E_{\delta}=0) \geq& \frac{1}{n}   H(Z^n|U^n,W^n)- \frac{1}{n}   - \log_2 |\mc{U}| \cdot \PP_{\sigma}\big(E_{\delta}=1 \big) \label{eq:Control4}\\
 =& H(Z|U,W) - \frac{1}{n}   - \log_2 |\mc{U}| \cdot \PP_{\sigma}\big(E_{\delta}=1 \big).  \label{eq:Control5}
\end{align}
Equation \eqref{eq:Control1} comes from the i.i.d. source and Fano's inequality.\\
Equation \eqref{eq:Control2} comes from the cardinality of codebook given by \eqref{eq:AchievabilityB1}. This argument is also used in \cite[Eq. (23)]{MerhavShamai(StateMasking)07}.\\
Equation \eqref{eq:Control3} comes from the cardinality of $A_{\delta}(z^n|w^n)$, see also \cite[pp. 27]{ElGammalKim(book)11}.\\
Equation \eqref{eq:Control4} comes from Fano's inequality.\\
Equation \eqref{eq:Control4} comes from $ H(Z^n|U^n,W^n)= H(Z^n|U^n)=H(Z|U)=H(Z|U,W)$ due to the Markov chain $Z^n -\!\!\!\!\minuso\!\!\!\!- U^n -\!\!\!\!\minuso\!\!\!\!- W^n$ of the encoding $\sigma$, the i.i.d. property of the source $(U,Z)$, the Markov chain $Z -\!\!\!\!\minuso\!\!\!\!- U -\!\!\!\!\minuso\!\!\!\!- W$ of the single-letter characterization $\QQ_{UZW}\in\Q_0$.

Equations \eqref{eq:Beliefs5}-\eqref{eq:Control4} shows that on average, the posterior beliefs $ \PP_{\sigma,U_t}(\cdot|y^n,z^n,E_{\delta}=0)$ induced by strategy $\sigma$ is close to the target probability distribution $\QQ_U(\cdot|w,z)$.
\begin{align}
&  \E_{\sigma} \Bigg[ \frac{1}{n}  \sum_{t=1}^n D\bigg(  \PP_{\sigma,U_t}(\cdot|Y^n,Z^n,E_{\delta}=0) \bigg| \bigg|   \QQ_{U_t}(\cdot|W_t,Z_t) \bigg)\Bigg] \nonumber \\
\leq&2\delta + \eta + \frac{2}{n}   + 2 \log_2 |\mc{U}| \cdot \PP_{\sigma}\big(E_{\delta}=1 \big) := \epsilon. \label{eq:ControlFinal} 
\end{align}

Then we have: 
\begin{align}
\PP_\sigma(B^c_{\alpha,\gamma,\delta})=&1 - \PP_\sigma(B_{\alpha,\gamma,\delta})   \nonumber\\
=&\PP_\sigma(E_{\delta}=1) \PP_\sigma(B^c_{\alpha,\gamma,\delta}| E_{\delta}=1)  + \PP_\sigma(E_{\delta}=0) \PP_\sigma(B^c_{\alpha,\gamma,\delta}| E_{\delta}=0) \nonumber\\
\leq&\PP_\sigma(E_{\delta}=1)   +  \PP_\sigma(B^c_{\alpha,\gamma,\delta}| E_{\delta}=0) \nonumber\\
\leq&\varepsilon_2  +  \PP_\sigma(B^c_{\alpha,\gamma,\delta}| E_{\delta}=0) .\label{eq:ErrorTerm2}
\end{align}
Moreover:
\begin{align}
& \PP_\sigma(B^c_{\alpha,\gamma,\delta}| E_{\delta}=0)\nonumber\\
=&\sum_{w^n,y^n,z^n}\PP_{\sigma}\Big( (w^n,y^n,z^n)\in  B^c_{\alpha,\gamma,\delta} \Big| E_{\delta}=0\Big)  \label{eq:MarkovIneqB0} \\
=&\sum_{w^n,y^n,z^n}\PP_{\sigma}\Bigg( (w^n,y^n,z^n)\,\quad \text{ s.t. } \quad  \frac{|T_\alpha(w^n,y^n,z^n)|}{n}< 1-\gamma  \Bigg| E_{\delta}=0\Bigg)  \label{eq:MarkovIneqB1} \\
=& \PP_{\sigma}\Bigg( \frac{1}{n} \cdot \bigg|\bigg\{t , \text{ s.t. } D\Big(\PP_{\sigma,U_t}(\cdot|y^n,z^n)\Big|\Big|\QQ_{U_t}(\cdot|w_t,z_t)\Big)\leq  \frac{\alpha^2}{2\ln 2}   \bigg\}\bigg| < 1 -\gamma \Bigg| E_{\delta}=0 \Bigg) \label{eq:MarkovIneqB2}\\
=& \PP_{\sigma}\Bigg( \frac{1}{n} \cdot \bigg| \bigg\{ t , \text{ s.t. } D\Big(\PP_{\sigma,U_t}(\cdot|y^n,z^n)\Big|\Big|\QQ_{U_t}(\cdot|w_t,z_t)\Big)>  \frac{\alpha^2}{2\ln 2}   \bigg\}\bigg| \geq \gamma \Bigg| E_{\delta}=0 \Bigg)  \label{eq:MarkovIneqB2b}\\
\leq& \frac{2\ln 2}{\alpha^2\gamma}   \cdot \E_{\sigma}\bigg[  \frac{1}{n}  \sum_{t=1}^n   D\Big(\PP_{\sigma,U_t}(\cdot|y^n,z^n)\Big|\Big|\QQ_{U_t}(\cdot|w_t,z_t)\Big)\bigg] \label{eq:MarkovIneqB3}  \\
\leq&\frac{2\ln 2}{\alpha^2\gamma}   \cdot \bigg( \eta +  \delta + \frac{2}{n}   +2 \log_2 |\mc{U}| \cdot \PP_{\sigma}\big(E_{\delta}=1 \big)  \bigg) \label{eq:MarkovIneqB4}.
\end{align}
Equation \eqref{eq:MarkovIneqB0} to \eqref{eq:MarkovIneqB2b} are simple reformulations.\\
Equation \eqref{eq:MarkovIneqB3} comes from the double use of Markov's inequality as in \cite[Lemma A.21, pp. 42]{LeTreustTomala19}. \\
Equation \eqref{eq:MarkovIneqB4} comes from \eqref{eq:ControlFinal}.

Combining equations \eqref{eq:BoundError0}, \eqref{eq:ErrorTerm2}, \eqref{eq:MarkovIneqB4} and choosing $\eta>0$ small, we obtain the following statement:
\begin{align}
&\forall \varepsilon>0, \;\forall \alpha>0, \;\forall \gamma>0,\;\exists \bar{\delta}>0,\;\forall \delta< \bar{\delta}, \;\exists \bar{n}\in \N^{\star},\;\forall n\geq  \bar{n},  \exists \sigma, \text{ s.t. } \PP_{\sigma}(B_{\alpha,\gamma,\delta}^c) \leq \varepsilon.\label{eq:propWynerZivCoding}
\end{align}
This concludes the proof of Proposition \ref{prop:WynerZivCoding}.




\section{Converse proof of Theorem \ref{theo:MaxMinStackelberg}}\label{sec:ConverseProof}

We consider an encoding strategy $\sigma$ of length $n\in\N^{\star}$. We denote by $T$ the uniform random variable over $\{1,\ldots,n\}$ and let $Z^{-T}$ stand for $(Z_1,\ldots,Z_{T-1},Z_{T+1},\ldots Z_n)$, where $Z_T$ has been removed. We identify $(U,Z)=(U_T,Z_T)$ and we introduce the auxiliary random variable $W=(Y^n,Z^{-T},T)$ whose joint probability distribution $\PP_{UZW}$ is defined by
\begin{align}
\PP(u,z,w) =& \PP_{\sigma}\big(u_T,z_T,y^n,z^{-T},T\big) \nonumber \\
=& \PP(T=t) \cdot  \PP_{\sigma}\big(u_T,z_T,y^n,z^{-T}\big|T=t\big) \nonumber \\
=& \frac{1}{n} \cdot  \PP_{\sigma}\big(u_t, z_t,y^n,z^{-t}\big) ,\quad \forall (u,w,z,u^n,z^n,y^n). \label{eq:distributionW}
\end{align}  
This ensures that the Markov chain $W -\!\!\!\!\minuso\!\!\!\!- U -\!\!\!\!\minuso\!\!\!\!- Z$ is satisfied. Let us fix a decoding strategy $\tau_{V^n|Y^nZ^n}$ and define $\tilde{\tau}_{V|WZ} = \tilde{\tau}_{V|Y^nZ^{-T}TZ} = \tau_{V_T|Y^nZ^n}$. The encoder's long-run distortion writes:
\begin{align}
d^{\,n}_{\textsf{e}}(\sigma,\tau)
=& \sum_{u^n,z^n,y^n}\PP_{\sigma}(u^n,z^n,y^n )  \sum_{v^n} \tau(v^n| y^n,z^n )  \cdot \Bigg[    \frac{1}{n} \sum_{t=1}^n d_{\textsf{e}}(u_t,z_t,v_t)\Bigg] \label{eq:Reformulation2} \\
=& \sum_{t=1}^n  \sum_{u_t,z_t,\atop z^{-t},y^n}  \frac{1}{n}  \cdot \PP_{\sigma}(u_t,z^n,y^n )  \sum_{v_t}\tau(v_t| y^n ,z^n)  \cdot     d_{\textsf{e}}(u_t,z_t,v_t)\label{eq:Reformulation3} \\
=& \sum_{u_t,z_t,y^n,\atop z^{-t},t} \PP_{\sigma}(u_t,z_t,y^n,z^{-t},t)  \sum_{v_t}\tau(v_t| z_t, y^n,z^{-t},t )  \cdot     d_{\textsf{e}}(u_t,z_t,v_t) \label{eq:Reformulation4} \\
=& \sum_{u,z,w} \PP(u,z,w )  \sum_{v}\tilde{\tau}(v| w,z ) \cdot     d_{\textsf{e}}(u,z,v).\label{eq:Reformulation5}
\end{align}
Equations \eqref{eq:Reformulation2} - \eqref{eq:Reformulation4} are reformulations and re-orderings.\\
Equation \eqref{eq:Reformulation5} comes from replacing the random variables $(Y^{n},Z^{-T},T)$ by $W$ and $(U_T,Z_T)$ by $(U,Z)$, whose distribution is defined in \eqref{eq:distributionW}.

Equations \eqref{eq:Reformulation2} - \eqref{eq:Reformulation5} are also valid for the decoder's distortion $d^{\,n}_{\textsf{d}}(\sigma,\tau)= \sum_{u,z,\atop w,v} \PP(u,z,w )\tilde{\tau}(v| w,z ) \cdot     d_{\textsf{d}}(u,z,v)$. A best-reply strategy $\tau\in \textsf{BR}_{\textsf{d}}(\sigma)$ reformulates as:
\begin{align}
& \tau \in \argmin_{\tau'_{V^n|Y^nZ^n}} \sum_{u^n,z^n,\atop x^n,y^n,v^n} \PP_{\sigma}(u^n,z^n,x^n,y^n) \cdot \tau'(v^n|y^n,z^n)\cdot  \Bigg[ \frac{1}{n} \sum_{t=1}^n d_{\textsf{d}}(u_t,z_t,v_t)\Bigg]\\
\Longleftrightarrow&\tilde{\tau}_{V|WZ} \in \argmin_{\tilde{\tau}'_{V|WZ}}  \sum_{u,z, w} \PP(u,z,w) \cdot \tilde{\tau}'(v|w,z) \cdot   d_{\textsf{d}}(u,z,v)\\
\Longleftrightarrow&\tilde{\tau}_{V|WZ} \in \Q_2\big(\PP_{UZW}\big).\label{eq:Identification}
\end{align}  
We now prove that the distribution $\PP_{UZW}$ defined in \eqref{eq:distributionW}, satisfies the information constraint of the set ${\Q}_0$.
\begin{align}
0 \leq& I(X^n ; Y^n) - I(U^n,Z^n;Y^n) \label{eq:ConverseW1} \\
\leq& \sum_{t=1}^n H( Y_t)  -  \sum_{t=1}^n H(Y_t | X_t) -    I(U^n;Y^n | Z^n)   \label{eq:ConverseW2} \\
\leq& n \cdot \max_{\PP_X} I(X ; Y)  -    \sum_{t=1}^n   I(U_t ;Y^n | Z^n,U^{t-1}) \label{eq:ConverseW3} \\
=& n \cdot \max_{\PP_X} I(X ; Y)  -    \sum_{t=1}^n   I(U_t;Y^n, Z^{-t},U^{t-1} | Z_t) \label{eq:ConverseW4} \\
\leq& n \cdot \max_{\PP_X} I(X ; Y)  -    \sum_{t=1}^n   I(U_t;Y^n, Z^{-t} | Z_t) \label{eq:ConverseW5} \\
=& n \cdot \max_{\PP_X} I(X ; Y)  -   n \cdot     I(U_T;Y^n, Z^{-T} | Z_T,T) \label{eq:ConverseW6} \\
=& n \cdot \max_{\PP_X} I(X ; Y)  -   n \cdot     I(U_T;Y^n, Z^{-T},T | Z_T) \label{eq:ConverseW7} \\
=& n \cdot \max_{\PP_X} I(X ; Y)  -   n \cdot     I(U;W | Z) \label{eq:ConverseW8} \\
=& n \cdot  \bigg(\max_{\PP_X} I(X ; Y)  -  I(U ; W) + I(Z ; W) \bigg) \label{eq:ConverseW9} .
\end{align}
Equation \eqref{eq:ConverseW1} comes from the Markov chain $Y^n  -\!\!\!\!\minuso\!\!\!\!- X^n   -\!\!\!\!\minuso\!\!\!\!- (U^n,Z^n)$.\\
Equation \eqref{eq:ConverseW2} comes from the memoryless property of the channel and from removing the positive term $I(U^n; Z^n)\geq0$.\\
Equation \eqref{eq:ConverseW3} comes from taking the maximum over $\PP_X$ and the chain rule.\\
Equation \eqref{eq:ConverseW4} comes from the i.i.d. property of the source $(U,Z)$ which implies $I(U_t,Z_t;Z^{-t},U^{t-1})=I(U_t;Z^{-t},U^{t-1} | Z_t)=0$.\\
Equation \eqref{eq:ConverseW5} comes from removing $I(U_t;U^{t-1} | Y^n, Z^{-t},Z_t)\geq0$.\\
Equation \eqref{eq:ConverseW6} comes from the introduction of the uniform random variable $T\in\{1,\ldots,n\}$.\\
Equation \eqref{eq:ConverseW7} comes from the independence between $T$ and $(U_T,Z_T)$, which implies $I(U_T,Z_T;T)  = I(U_T;T|Z_T)  =0$.\\
Equation \eqref{eq:ConverseW8} comes from the identification of $(U,Z)=(U_T,Z_T)$ and $W = (Y^{n},Z^{-T},T)$. \\
Equation \eqref{eq:ConverseW9} comes from the  Markov chain property $W -\!\!\!\!\minuso\!\!\!\!- U -\!\!\!\!\minuso\!\!\!\!- Z$. This proves that the distribution $\PP_{\sigma,UZW}$  belongs to the set $\Q_0$.

Therefore, for any encoding strategy $\sigma$ and all $n$, we have:
\begin{align}
&\max_{\tau \in \textsf{BR}_{\textsf{d}}(\sigma)}  d_{\textsf{e}}^{\,n}(\sigma,\tau)  \\
=&\max_{\tilde{\tau}_{V|WZ} \in \atop \Q_2(\PP_{UZW})}   \sum_{u,z,w} \PP(u,z,w )  \sum_{v}\tilde{\tau}(v| w,z ) \cdot     d_{\textsf{e}}(u,z,v)\\
=&\max_{\tilde{\tau}_{V|WZ} \in  \atop \Q_2(\PP_{UZW})} \E_{\PP(u,z,w) \atop  \tilde{\tau}_{V|WZ} } \bigg[d_{\textsf{e}}(U,Z,V)\bigg]\\
\geq&\inf_{ \QQ_{UZW} \in \Q_0} \max_{\QQ_{V|WZ} \in  \atop \Q_2(\QQ_{UZW})} \E_{\QQ_{UZW} \atop  \QQ_{V|WZ}} \bigg[d_{\textsf{e}}(U,Z,V)\bigg]= D_{\textsf{e}}^{\star}.
\end{align}
The last inequality holds because  $\PP_{\sigma,UZW}\in {\Q}_0$.

The proof for the cardinality bound $|\mc{W}| =  \min(|\mc{U}|+1,|\mc{V}|^{|\mc{Z}|})$ follows two arguments. The bound $|\mc{W}| =  |\mc{U}|+1$ comes \cite[Corollary 17.1.5, pp. 157]{rockafellar1970convex}, also in \cite[Corollary A.2, pp. 26]{LeTreustTomala19}. The bound $|\mc{W}| =  |\mc{V}|^{|\mc{Z}|}$ comes from assuming that the encoder tells to the decoder to select a function from the side information $\mc{Z}$ to the symbols $\mc{V}$, as discussed in Sec. \ref{sec:Concavification}. 

This concludes the proof of \eqref{eq:Converse} in Theorem \ref{theo:MaxMinStackelberg}.




\section{Proof of Proposition \ref{prop:WynerZivDSBS}}\label{sec:ProofPropDSBS}
The average distortion $\Psi_{\textsf{e}}(q)$ defined in \eqref{eq:AverageFunctionPsi1} is piece-wise linear, hence the optimal triple of posteriors may belong to distinct intervals $q_1\in [0, \nu_1)$, $q_2\in [\nu_1,\nu_2)$, $q_3\in [\nu_2,1]$. Since $\delta_0=\delta_1$ and $\kappa=0$, the function $\Psi_{\textsf{e}}(q)$ is symmetric and constant over the interval $[\nu_1,\nu_2]$, as depicted in Fig \ref{fig:DSBS_02}. The optimal splitting must satisfy $(q_1,q_2,q_3)=(q_1,\frac12,1- q_1)$, since $q_2=\frac12$ provides the highest entropy $h(q_2)=H(U|Z)$. Equations \eqref{eq:Lambda1}-\eqref{eq:Lambda3} reformulate as
\begin{align}
\lambda_1 =& \frac12 \cdot \frac{C}{H(U|Z) -h(q_1)},\\
\lambda_2 =& 1- \frac{C}{H(U|Z) -h(q_1)},\\
\lambda_3 =& \frac12 \cdot \frac{C}{H(U|Z) -h(q_1)}.
\end{align}
We examine the feasibility conditions, i.e. $(\lambda_1,\lambda_2,\lambda_3)\in [0,1]^3$. We recall that the notation $h^{-1}\big(H(U|Z)-C\big)$ stands for the unique solution $q\in[0,\frac12]$ of the equation $h(q)=H(U|Z)-C$. Since $H(U|Z)\geq h(q_1)$ for all $q_1$, we have $\lambda_1\geq 0$, $\lambda_2\leq 1$ and $\lambda_3\geq 0$, moreover
\begin{align}
\lambda_1\leq 1 
\Longleftrightarrow   h(q_1) \leq   H(U|Z)  - \frac12 \cdot C 
\Longleftrightarrow   q_1 \leq   h^{-1}\Big(H(U|Z)  - \frac12 \cdot C\Big)
\end{align}
\begin{align}
\lambda_2\geq 0 
\Longleftrightarrow h(q_1) \leq H(U|Z) -C
\Longleftrightarrow q_1 \leq h^{-1}\Big(H(U|Z) -C\Big).
\end{align}
Since the function $h^{-1}$ is increasing over the interval $[0,1]$, we have $h^{-1}\Big(H(U|Z) -\frac12 \cdot C\Big)\leq h^{-1}\Big(H(U|Z) -C\Big)$. This proves the following Lemma.
\begin{lemma}\label{lemma:FeasiblePosterior}
The splitting $(q_1,q_2,q_3)=(q_1,\frac12,1-q_1)$ is feasible if and only if 
\begin{align}
q_1 \leq h^{-1}\Big(H(U|Z) -C\Big).
\end{align}
\end{lemma}

We assume that $(q_1,q_2,q_3)=(q_1,\frac12,1- q_1)$, we define the \emph{encoder's distortion function} by
\begin{align}
\Phi_{\textsf{e}}(q_1)=&\lambda_1 \Psi_{\textsf{e}}(q_1) + \lambda_2 \Psi_{\textsf{e}}\Big(\frac12\Big) + \lambda_3 \Psi_{\textsf{e}}(1-q_1)\\
=&  \frac12 \cdot \frac{C}{H(U|Z) -h(q_1)}\cdot   q_1 + \bigg( 1 - \frac{C}{H(U|Z) -h(q_1)}\bigg)\cdot \delta +  \frac12 \cdot \frac{C}{H(U|Z) -h(q_1)} \cdot q_1\\
=&\delta - \frac{( \delta - q_1) \cdot  C }{H(U|Z) -h(q_1)},
\end{align}
Its derivative is
\begin{align}
\Phi_{\textsf{e}}'(q_1)
=&  \frac{C}{(H(U|Z) -h(q_1))^2} \cdot \bigg(   H(U|Z) -h(q_1) - h'(q_1)\cdot ( \delta - q_1)\bigg),
\end{align}
where the derivative of the entropy $h'(q)$ writes
\begin{align}
h'(q) =& \log_2 \frac{1-q}{q} - (1 - 2\cdot \delta)\cdot \log_2\frac{1-q\star \delta}{q\star \delta}.
\end{align}
We examine the term $k(q)= \Big(H(U|Z) -h(q) - h'(q)\cdot ( \delta - q)\Big)$ of the function $\Phi_{\textsf{e}}'(q)$. We have $\lim_{q\to 0}k(q) = - \infty$ since $\lim_{q\to0}h'(q) = +\infty$ and $k(\delta)= H(U|Z) -h(\delta) = H_b(\delta\star\delta) - H_b(\delta)>0$ since $\delta<\frac12$. The derivative $k'(q)= - h''(q)\cdot ( \delta - q)>0$ is strictly positive because $ \delta > q$ and the entropy is strictly concave $h''(q)<0$. Hence the equation $k(q)=0$ has a unique solution $q^{\star}\in(0,\delta[$. 

The derivative $\Phi_{\textsf{e}}'(q)$ is non-positive on the interval $q\in (0,q^{\star}]$ and non-negative on the interval $q\in [q^{\star},\delta)$, hence the distortion $\Phi_{\textsf{e}}(q)$ reaches its minimum in $q^{\star}$.

1) If $q^{\star}\leq h^{-1}\Big(H(U|Z) -C\Big)$, then the optimal splitting is \\\\
\begin{tabular}{|l|l|l|}
\hline
$q_1 = q^{\star}$ & $q_2 = \frac12 $ & $q_3 = 1-q^{\star}$\\
\hline
$\lambda_1=  \frac12 \cdot \frac{C}{H(U|Z) - h(q^{\star})}$& $\lambda_2=  1- \frac{C}{H(U|Z) - h(q^{\star})}$& $\lambda_3 =  \frac12 \cdot \frac{C}{H(U|Z) - h(q^{\star})}$\\
\hline
\end{tabular}\\\\
corresponding to the strategies\\\\
\begin{tabular}{|l|l|l|}
\hline
$\alpha_1 =  (1 - q^{\star}) \cdot \frac{C}{H(U|Z) - h(q^{\star})}$& $\alpha_2 =  1- \frac{C}{H(U|Z) - h(q^{\star})}$& $\alpha_3 =  q^{\star} \cdot \frac{C}{H(U|Z) - h(q^{\star})}$\\
\hline
$\beta_1 =  q^{\star}  \cdot \frac{C}{H(U|Z) - h(q^{\star})}$& $\beta_2 =  1- \frac{C}{H(U|Z) - h(q^{\star})}$& $\beta_3 =  (1 - q^{\star}) \cdot \frac{C}{H(U|Z) - h(q^{\star})}$\\
\hline
\end{tabular}\\\\
and the optimal distortion is
\begin{align}
D_{\textsf{e}}^{\star} =& \delta - C\cdot \frac{ \delta - q^{\star}  }{H(U|Z) - h(q^{\star})} .
\end{align}

2) If $q^{\star}> h^{-1}\Big(H(U|Z) -C\Big)$, then the optimal splitting is\\\\
\begin{tabular}{|l|l|l|}
\hline
$q_1 = h^{-1}\big(H(U|Z)-C\big)$ & $q_2 = \frac12 $ & $q_3 = 1-h^{-1}\big(H(U|Z)-C\big)$\\
\hline
$\lambda_1=  \frac12 $& $\lambda_2=  0$& $\lambda_3 =  \frac12$\\
\hline
\end{tabular}\\\\
corresponding to the strategies\\\\
\begin{tabular}{|l|l|l|}
\hline
$\alpha_1 =  1 - h^{-1}\big(H(U|Z)-C\big)$& $\alpha_2 =  0$& $\alpha_3 =  h^{-1}\big(H(U|Z)-C\big)$\\
\hline
$\beta_1 =  h^{-1}\big(H(U|Z)-C\big)$& $\beta_2 =  0$& $\beta_3 =  1 - h^{-1}\big(H(U|Z)-C\big)$\\
\hline\hline
\end{tabular}\\\\
and the optimal distortion is
\begin{align}
D_{\textsf{e}}^{\star} =& h^{-1}\big(H(U|Z)-C\big).
\end{align}

3) If $C>H(U|Z)$, then the extreme splitting $(q_1,q_3)=(0,1)$ is feasible and $D_{\textsf{e}}^{\star} =0$.

This concludes the proof of Proposition  \ref{prop:WynerZivDSBS}.
%

\begin{thebibliography}{10}
\providecommand{\url}[1]{#1}
\csname url@samestyle\endcsname
\providecommand{\newblock}{\relax}
\providecommand{\bibinfo}[2]{#2}
\providecommand{\BIBentrySTDinterwordspacing}{\spaceskip=0pt\relax}
\providecommand{\BIBentryALTinterwordstretchfactor}{4}
\providecommand{\BIBentryALTinterwordspacing}{\spaceskip=\fontdimen2\font plus
\BIBentryALTinterwordstretchfactor\fontdimen3\font minus
  \fontdimen4\font\relax}
\providecommand{\BIBforeignlanguage}[2]{{%
\expandafter\ifx\csname l@#1\endcsname\relax
\typeout{** WARNING: IEEEtran.bst: No hyphenation pattern has been}%
\typeout{** loaded for the language `#1'. Using the pattern for}%
\typeout{** the default language instead.}%
\else
\language=\csname l@#1\endcsname
\fi
#2}}
\providecommand{\BIBdecl}{\relax}
\BIBdecl

\bibitem{LeTreustTomala(Allerton)16}
M.~Le~Treust and T.~Tomala, ``Information design for strategic coordination of
  autonomous devices with non-aligned utilities,'' in \emph{Proc. 54th Annual Allerton Conference on Communication, Control, and Computing (Allerton)}, Monticello, Illinois, Sept. 2016,  pp. 233--242.

\bibitem{LeTreustTomala(Gretsi)17}
------, ``Persuasion bayésienne pour la coordination stratégique d'appareils
  autonomes ayant des objectifs non-alignés,'' in \emph{Actes de la Conférence
  du Groupement de Recherche en Traitement du Signal et des Images (GRETSI17)},
  Juan-les-Pins, France, Sept. 2017.

\bibitem{LeTreustTomala(IZS)18}
------, ``Strategic coordination with state information at the decoder,'' in
  \emph{Proc. International Zurich Seminar on Information and Communication (IZS18)}, Zurich, Switzerland, Feb. 2018, pp. 30--34.

\bibitem{LeTreustTomala19}
------, ``Persuasion with limited communication capacity,'' \emph{Journal of Economic
  Theory}, vol. 184, p. 104940, 2019.

\bibitem{wyner-it-1976}
A.~D. Wyner and J.~Ziv, ``The rate-distortion function for source coding with
  side information at the decoder,'' \emph{IEEE Transactions on Information
  Theory}, vol.~22, no.~1, pp. 1--11, Jan. 1976.

\bibitem{MerhavShamai03}
N.~Merhav and S.~Shamai, ``On joint source-channel coding for the {W}yner-{Z}iv
  source and the {G}el'fand-{P}insker channel,'' \emph{IEEE Transactions on
  Information Theory}, vol.~49, no.~11, pp. 2844--2855, Nov. 2003.

\bibitem{KamenicaGentzkow11}
E.~Kamenica and M.~Gentzkow, ``Bayesian persuasion,'' \emph{American Economic
  Review}, vol. 101, pp. 2590--2615, 2011.

\bibitem{AkyolLangbortBasar15}
E.~Akyol, C.~Langbort, and T.~Ba\c{s}ar, ``Strategic compression and
  transmission of information,'' in \emph{Proc. IEEE Information Theory Workshop -
  Fall (ITW)}, Jeju, South Korea, Oct. 2015, pp. 219--223.

\bibitem{AkyolLangbortBasar16}
------, ``On the role of side information in strategic communication,'' in
  \emph{IEEE International Symposium on Information Theory (ISIT)}, Barcelona, Spain, July 2016,  pp. 1626--1630.

\bibitem{AkyolLangbortBasarIEEE17}
------, ``Information-theoretic approach to strategic communication as a
  hierarchical game,'' \emph{Proceedings of the IEEE}, vol. 105, no.~2, pp.
  205--218, 2017.

\bibitem{CrawfordSobel1982StrategicInformation}
V.~P. Crawford and J.~Sobel, ``Strategic information transmission,''
  \emph{Econometrica}, vol.~50, no.~6, pp. 1431--1451, 1982.

\bibitem{NadendlaLangbortBasar(TCOM)18}
V.~S.~S. {Nadendla}, C.~{Langbort}, and T.~{Ba\c{s}ar}, ``Effects of subjective
  biases on strategic information transmission,'' \emph{IEEE Transactions on
  Communications}, vol.~66, no.~12, pp. 6040--6049, Dec. 2018.

\bibitem{SaritasYukselGezici(ArXiV)17}
S.~Sar{\i}ta\c{s}, S.~Yüksel, and S.~Gezici, ``Dynamic signaling games with
  quadratic criteria under nash and stackelberg equilibria,'' \emph{preprint
  available [on-line] https://arxiv.org/abs/1704.03816}, 2017.

\bibitem{SaritasYukselGezici(TAC)17}
------, ``Quadratic multi-dimensional signaling games and affine equilibria,''
  \emph{IEEE Transactions on Automatic Control}, vol.~62, no.~2, pp. 605--619,
  Feb. 2017.

\bibitem{SaritasGeziciYuksel(TSP)19}
S.~{Sar{\i}ta\c{s}}, S.~{Gezici}, and S.~{Yüksel}, ``Hypothesis testing under
  subjective priors and costs as a signaling game,'' \emph{IEEE Transactions on
  Signal Processing}, vol.~67, no.~19, pp. 5169--5183, Oct. 2019.

\bibitem{FarokhiTeixeiraLangbort17}
F.~Farokhi, A.~M.~H. Teixeira, and C.~Langbort, ``Estimation with strategic
  sensors,'' \emph{IEEE Transactions on Automatic Control}, vol.~62, no.~2, pp.
  724--739, Feb. 2017.

\bibitem{MarecekShortenYu15}
J.~Y.~Y. Jakub~Mare\v{c}ek, Robert~Shorten, ``Signaling and obfuscation for
  congestion control,'' \emph{ACM SIGecom Exchanges}, vol.~88, no.~10, pp.
  2086--2096, 2015.

  \bibitem{TavafoghiTeneketzis(Allerton)17}
H.~{Tavafoghi} and D.~{Teneketzis}, ``Informational incentives for congestion
  games,'' in   \emph{Proc. 55th Annual Allerton Conference on Communication,
  Control, and Computing (Allerton)}, Monticello, Illinois, Oct. 2017, pp. 1285--1292.

\bibitem{DughmiXu16}
S.~Dughmi and H.~Xu, ``Algorithmic bayesian persuasion,'' in \emph{Proc. 47th ACM Symposium on Theory of Computing (STOC)}, Portland, Oregon, June 2016.

\bibitem{Dughmi17}
S.~Dughmi, ``Algorithmic information structure design: A survey,'' \emph{ACM
  SIGecom Exchanges}, vol.~15, no.~2, pp. 2--24, Jan. 2017.

\bibitem{DughmiKempeQiang16}
S.~Dughmi, D.~Kempe, and R.~Qiang, ``Persuasion with limited communication,'' in
  \emph{Proc. 17th ACM conference on Economics and Computation
  (ACM EC'16)}, Maastricht, The Netherlands, June 2016.

\bibitem{BerryTse(ShannonMetNash)11}
R.~Berry and D.~Tse, ``{S}hannon meets {N}ash on the interference channel,''
  \emph{IEEE Transactions on Information Theory}, vol.~57, no.~5, pp.
  2821--2836, May 2011.

\bibitem{PerlazaTandonPoorHan15}
S.~M. Perlaza, R.~Tandon, H.~V. Poor, and Z.~Han, ``Perfect output feedback in
  the two-user decentralized interference channel,'' \emph{IEEE Transactions on
  Information Theory}, vol.~61, no.~10, pp. 5441--5462, Oct. 2015.

\bibitem{GamalCover82}
A.~E. Gamal and T.~Cover, ``Achievable rates for multiple descriptions,''
  \emph{IEEE Transactions on Information Theory}, vol.~28, no.~6, pp. 851--857,
  Nov. 1982.

\bibitem{LapidothMalarWigger14}
A.~Lapidoth, A.~Malär, and M.~Wigger, ``Constrained source-coding with side
  information,'' \emph{IEEE Transactions on Information Theory}, vol.~60,
  no.~6, pp. 3218--3237, June 2014.

\bibitem{DemboWeissman03}
A.~{Dembo} and T.~{Weissman}, ``The minimax distortion redundancy in noisy
  source coding,'' \emph{IEEE Transactions on Information Theory}, vol.~49,
  no.~11, pp. 3020--3030, Nov. 2003.

\bibitem{TimoGrantKramer13}
R.~{Timo}, A.~{Grant}, and G.~{Kramer}, ``Lossy broadcasting with complementary
  side information,'' \emph{IEEE Transactions on Information Theory}, vol.~59,
  no.~1, pp. 104--131, Jan. 2013.

\bibitem{Yamamoto88}
H.~{Yamamoto}, ``A rate-distortion problem for a communication system with a
  secondary decoder to be hindered,'' \emph{IEEE Transactions on Information
  Theory}, vol.~34, no.~4, pp. 835--842, July 1988.

\bibitem{Yamamoto97}
H.~Yamamoto, ``Rate-distortion theory for the shannon cipher system,''
  \emph{IEEE Transactions on Information Theory}, vol.~43, no.~3, pp. 827--835,
  May 1997.

\bibitem{SchielerCuff(RateDistortion14)}
C.~Schieler and P.~Cuff, ``Rate-distortion theory for secrecy systems,''
  \emph{IEEE Transactions on Information Theory}, vol.~60, no.~12, pp.
  7584--7605, Dec. 2014.

\bibitem{SchielerCuff(Henchman)16}
------, ``The henchman problem: Measuring secrecy by the minimum distortion in
  a list,'' \emph{IEEE Transactions on Information Theory}, vol.~62, no.~6, pp.
  3436--3450, June 2016.

\bibitem{Lapidoth97}
A.~Lapidoth, ``On the role of mismatch in rate distortion theory,'' \emph{IEEE
  Transactions on Information Theory}, vol.~43, no.~1, pp. 38--47, Jan. 1997.

\bibitem{MerhavKaplanLapidothShamai94}
N.~{Merhav}, G.~{Kaplan}, A.~{Lapidoth}, and S.~{Shamai Shitz}, ``On
  information rates for mismatched decoders,'' \emph{IEEE Transactions on
  Information Theory}, vol.~40, no.~6, pp. 1953--1967, Nov. 1994.

\bibitem{Zamir02}
R.~{Zamir}, ``The index entropy of a mismatched codebook,'' \emph{IEEE
  Transactions on Information Theory}, vol.~48, no.~2, pp. 523--528, Feb. 2002.

\bibitem{Somekh15}
A.~{Somekh-Baruch}, ``A general formula for the mismatch capacity,'' \emph{IEEE
  Transactions on Information Theory}, vol.~61, no.~9, pp. 4554--4568, Sep.
  2015.

\bibitem{ScarlettMartinezFabregas16}
J.~{Scarlett}, A.~{Martinez}, and A.~{Guillén i Fàbregas}, ``Multiuser random
  coding techniques for mismatched decoding,'' \emph{IEEE Transactions on
  Information Theory}, vol.~62, no.~7, pp. 3950--3970, July 2016.

\bibitem{ZhouTanMotani19}
L.~{Zhou}, V.~Y.~F. {Tan}, and M.~{Motani}, ``The dispersion of mismatched
  joint source-channel coding for arbitrary sources and additive channels,''
  \emph{IEEE Transactions on Information Theory}, vol.~65, no.~4, pp.
  2234--2251, April 2019.

\bibitem{Forges94}
F.~Forges, ``Non-zero-sum repeated games and information transmission,''
  \emph{in: N. Meggido, Essays in Game Theory in Honor of Michael Maschler,
  Springer-Verlag}, no.~6, pp. 65--95, 1994.

\bibitem{stackelberg-book-1934}
H.~von Stackelberg, \emph{Marketform und Gleichgewicht}.\hskip 1em plus 0.5em
  minus 0.4em\relax Oxford University Press, 1934.

\bibitem{Nash51}
J.~Nash, ``Non-cooperative games,'' \emph{Annals of Mathematics}, vol.~54, pp.
  286--295, 1951.

\bibitem{BergemannMorris16}
D.~Bergemann and S.~Morris, ``Information design, bayesian persuasion, and
  bayes correlated equilibrium,'' \emph{American Economic Review Papers and
  Proceedings}, vol. 106, no.~5, pp. 586--591, May 2016.

\bibitem{Taneva16}
I.~Taneva, ``Information design,'' \emph{Manuscript, School of Economics, The
  University of Edinburgh}, 2016.

\bibitem{BergemannMorris17}
D.~Bergemann and S.~Morris, ``Information design: a unified perspective,''
  \emph{Cowles Foundation Discussion Paper No 2075}, 2017.

\bibitem{AlonsoCamara(JET)2016}
R.~Alonso and O.~C\^{a}mara, ``Bayesian persuasion with heterogeneous priors,''
  \emph{Journal of Economic Theory}, vol. 165, pp. 672--706, 2016.

\bibitem{LaclauRenou17}
M.~Laclau and L.~Renou, ``Public persuasion,'' \emph{working paper}, Feb. 2017.

\bibitem{TsakasTsakas2017}
E.~Tsakas and N.~Tsakas, ``Noisy persuasion,'' \emph{working paper}, 2017.

\bibitem{Blume}
A.~Blume, O.~J. Board, and K.~Kawamura, ``Noisy talk,'' \emph{Theoretical
  Economics}, vol.~2, pp. 395--440, July 2007.

\bibitem{HernandezVonStengel14}
P.~Hern\'{a}ndez and B.~von Stengel, ``Nash codes for noisy channels,''
  \emph{Operations Research}, vol.~62, no.~6, pp. 1221--1235, Nov. 2014.

\bibitem{Sims03}
C.~Sims, ``Implication of rational inattention,'' \emph{Journal of Monetary
  Economics}, vol.~50, no.~3, pp. 665--690, April 2003.

\bibitem{GentzkowKamenica14}
M.~Gentzkow and E.~Kamenica, ``Costly persuasion,'' \emph{American Economic
  Review}, vol. 104, pp. 457--462, 2014.

\bibitem{NeymanOkada99}
A.~Neyman and D.~Okada, ``Strategic entropy and complexity in repeated games,''
  \emph{Games and Economic Behavior}, vol.~29, no. 1--2, pp. 191--223, 1999.

\bibitem{NeymanOkada00}
------, ``Repeated games with bounded entropy,'' \emph{Games and Economic
  Behavior}, vol.~30, no.~2, pp. 228--247, 2000.

\bibitem{NeymanOkada09}
------, ``Growth of strategy sets, entropy, and nonstationary bounded recall,''
  \emph{Games and Economic Behavior}, vol.~66, no.~1, pp. 404--425, 2009.

\bibitem{GossnerVieille02}
O.~Gossner and N.~Vieille, ``How to play with a biased coin?'' \emph{Games and
  Economic Behavior}, vol.~41, no.~2, pp. 206--226, 2002.

\bibitem{GossnerTomala06}
O.~Gossner and T.~Tomala, ``Empirical distributions of beliefs under imperfect
  observation,'' \emph{Mathematics of Operation Research}, vol.~31, no.~1, pp.
  13--30, 2006.

\bibitem{GossnerTomala07}
------, ``Secret correlation in repeated games with imperfect monitoring,''
  \emph{Mathematics of Operation Research}, vol.~32, no.~2, pp. 413--424, 2007.

\bibitem{GossnerLarakiTomala09}
O.~Gossner, R.~Laraki, and T.~Tomala, ``Informationally optimal correlation,''
  \emph{Mathematical Programming}, vol. 116, no. 1-2, pp. 147--172, 2009.

\bibitem{GossnerHernandezNeyman06}
O.~Gossner, P.~Hern\'{a}ndez, and A.~Neyman, ``Optimal use of communication
  resources,'' \emph{Econometrica}, vol.~74, no.~6, pp. 1603--1636, 2006.

\bibitem{Cuff(ImplicitCoordination)11}
P.~Cuff and L.~Zhao, ``Coordination using implicit communication,'' in \emph{Proc. IEEE
  Information Theory Workshop (ITW)}, Paraty, Brazil, Oct. 2011, pp. 467--471.

\bibitem{CuffPermuterCover10}
P.~Cuff, H.~Permuter, and T.~Cover, ``Coordination capacity,''
  \emph{IEEE Transactions on Information Theory}, vol.~56, no.~9, pp.
  4181--4206, Sept. 2010.

  \bibitem{CuffSchieler11}
P.~Cuff and C.~Schieler, ``Hybrid codes needed for coordination over the
  point-to-point channel,'' in \emph{Proc. 49th Annual Allerton Conference on Communication, Control, and Computing (Allerton)}, Monticello, Illinois, Sept. 2011, pp. 235--239.

\bibitem{LeTreust(EmpiricalCoordination)17}
M.~Le~Treust, ``Joint empirical coordination of source and channel,''
  \emph{IEEE Transactions on Information Theory}, vol.~63, no.~8, pp.
  5087--5114, Aug. 2017.

\bibitem{CerviaLuzziLeTreustBloch(IT)18}
G.~Cervia, L.~Luzzi, M.~Le~Treust, and M.~R. Bloch, ``Strong coordination of
  signals and actions over noisy channels with two-sided state information,'' \emph{IEEE Transactions on Information Theory}, vol. 66, no. 8, pp. 4681-4708, Aug. 2020.

\bibitem{shannon-bell-1948}
C.~E. Shannon, ``A mathematical theory of communication,'' \emph{Bell System
  Technical Journal}, vol.~27, pp. 379--423, 1948.

\bibitem{LeTreust(ISIT-TwoSided)15}
M.~Le~Treust, ``Empirical coordination with two-sided state information and
  correlated source and state,'' in \emph{Proc. IEEE International Symposium on
  Information Theory (ISIT)}, Hong-Kong, July 2015, pp. 466--470.

\bibitem{LeTreustBloch(ISIT)16}
M.~Le~Treust and M.~Bloch, ``Empirical coordination, state masking and state
  amplification: Core of the decoder's knowledge,'' in \emph{Proc. IEEE International
  Symposium on Information Theory (ISIT)}, Barcelona, Spain, July 2016, pp. 895--899.

\bibitem{LeTreustBloch(StateLeakageCoordination)18}
------, ``State leakage and coordination of actions: Core of decoder's
  knowledge,'' \emph{preprint available [on-line]
  https://arxiv.org/abs/1812.07026}, Dec. 2018.

\bibitem{Fekete1923}
M.~Fekete, ``{\"U}ber die verteilung der wurzeln bei gewissen algebraischen
  gleichungen mit ganzzahligen koeffizienten,'' \emph{Mathematische
  Zeitschrift}, vol.~17, no.~1, pp. 228--249, Dec. 1923.

\bibitem{ElGammalKim(book)11}
A.~E. Gamal and Y.-H. Kim, \emph{Network Information Theory}.\hskip 1em plus
  0.5em minus 0.4em\relax Cambridge University Press, Dec. 2011.

\bibitem{AM95}
R.~Aumann and M.~Maschler, \emph{Repeated Games with Incomplete
  Information}.\hskip 1em plus 0.5em minus 0.4em\relax MIT Press, Cambrige, MA,
  1995.

\bibitem{rockafellar1970convex}
R.~Rockafellar, \emph{Convex Analysis}, ser. Princeton landmarks in mathematics
  and physics.

\bibitem{cover-book-2006}
T.~M. Cover and J.~A. Thomas, \emph{Elements of information theory}.\hskip 1em
  plus 0.5em minus 0.4em\relax New York: 2nd. Ed., Wiley-Interscience, 2006.

\bibitem{MerhavShamai(StateMasking)07}
N.~Merhav and S.~Shamai, ``Information rates subject to state masking,''
  \emph{IEEE Transactions on Information Theory}, vol.~53, no.~6, pp.
  2254--2261, June 2007.

\end{thebibliography}




\end{document}